\documentclass[accepted=2018-02-07]{quantumarticle}
\usepackage{amsfonts, amsmath,amssymb,amsthm,bm}
\usepackage{graphicx}
\usepackage{hyperref}

\usepackage{etoolbox}

\usepackage{ulem}
\usepackage[dvipsnames]{xcolor}

\newtheorem{theorem}{Theorem}

\newtheorem{lemma}[theorem]{Lemma}

\newtheorem{corollary}[theorem]{Corollary}

\usepackage[utf8]{inputenc}
\usepackage{braket} 
\newcommand\ketbra[2]{\ket{#1}\!\bra{#2}}
\DeclareMathOperator{\Tr}{Tr}

\newcommand{\R}{\mathbb{R}}
\newcommand{\C}{\mathbb{C}}

\newcommand{\tr}{{\rm tr}}

\theoremstyle{plain}

\theoremstyle{definition}

\renewcommand{\thesection}{\Roman{section}} 
\renewcommand{\thesubsection}{\thesection.\arabic{subsection}}
\makeatletter
\def\p@subsection{}\makeatother

\begin{document}
\title{Quantum Horn's lemma, finite heat baths, and the third law of thermodynamics}
\author{Jakob Scharlau}
\affiliation{Department of Theoretical Physics, University of Heidelberg, Heidelberg, Germany}
\author{Markus P. M\"uller}
\affiliation{Institute for Quantum Optics and Quantum Information, Austrian Academy of Sciences, Boltzmanngasse 3, A-1090 Vienna, Austria}
\affiliation{Department of Applied Mathematics, University of Western Ontario, London, ON N6A 5BY, Canada}
\affiliation{Department of Philosophy, University of Western Ontario, London, ON N6A 5BY, Canada}
\affiliation{Perimeter Institute for Theoretical Physics, Waterloo, ON N2L 2Y5, Canada}
\affiliation{Department of Theoretical Physics, University of Heidelberg, Heidelberg, Germany}

\date{February 2, 2018}

\begin{abstract}
Interactions of quantum systems with their environment play a crucial role in resource-theoretic approaches to thermodynamics in the microscopic regime. Here, we analyze the possible state transitions in the presence of ``small'' heat baths of bounded dimension and energy. We show that for operations on quantum systems with fully degenerate Hamiltonian (noisy operations), all possible state transitions can be realized exactly with a bath that is of the same size as the system or smaller, which proves a quantum version of Horn's lemma as conjectured by Bengtsson and \.{Z}yczkowski. On the other hand, if the system's Hamiltonian is not fully degenerate (thermal operations), we show that some possible transitions can only be performed with a heat bath that is unbounded in size and energy, which is an instance of the third law of thermodynamics. In both cases, we prove that quantum operations yield an advantage over classical ones for any given finite heat bath, by allowing a larger and more physically realistic set of state transitions.
\end{abstract}

\maketitle

\bigskip

\section{Introduction}
Thermodynamics is traditionally viewed as a statistical theory of a large number of particles. In the last few years, however, there has been a wave of activity in extending the realm of thermodynamics to microscopic quantum systems, motivated by ideas from quantum information theory~\cite{Janzing,HHO,Dahlsten,Lidia,HO2013,Resource2013,Aberg,SSP2014,Renes2014,Faist2015,Brandao2015,Nicole2015,Egloff,MMP2015,Gemmer,MOThirdLaw,Guryanova,LJR2015,LKJR2015,KLOJ2016,Cwiklinski}. One key idea is to formulate thermodynamics as a \emph{resource theory}~\cite{Resource2013,Horo4,Coecke,BrandaoGour,LKR2015}: by specifying and restricting to a set of experimental operations that an agent is allowed to implement ``for free'', one obtains a rigorous framework which allows to address typical thermodynamics questions even in physical scenarios that are far from the usual thermodynamic limit.

\begin{figure}[!hbt]
\begin{center}
\includegraphics[angle=0, width=.3\textwidth]{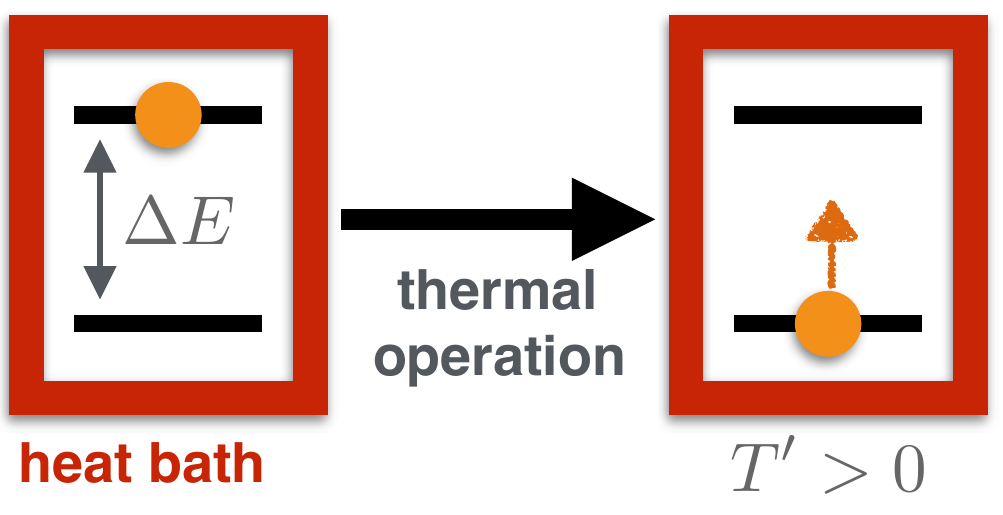}
\caption{Theorem~\ref{TheThird} can be seen as a simple instance of the third law of thermodynamics. We start with a qubit in its excited state, which serves on the one hand as a work reservoir, and on the other hand represents the system that we would like to cool. For any finite heat bath, thermal operations can never achieve the ground state exactly; there will always remain finite fluctuations, indicated by the arrow. The final temperature $T'$ will always be strictly positive, with a lower bound that depends inversely-extensively on the size and energy of the heat bath. The physical intuition of this will be elaborated in more detail in the conclusions.}
\label{fig_cooling}
\end{center}
\end{figure}

One way to formulate a suitable resource theory is by declaring the set of \emph{thermal operations} as allowed. A thermal operation consists of an arbitrary global energy-preserving unitary transformation on the given quantum system $A$ and an arbitrary heat bath $B$ (at fixed background temperature), and the subsequent disposal of the heat bath. By making $A$ itself consist of several parts (say, a single particle and a work-storage device), we can study the interplay of heat and work in microscopic systems. For example, we can ask how much work can be extracted with high probability from a given quantum state by using only thermal operations. Many of these questions have been answered explicitly, and have been shown to refine -- and converge to -- the usual thermodynamic results in the limit of a large number of weakly interacting particles.

Even though this approach has been tremendously successful in providing insights into thermodynamics in the microscopic regime, it has also raised questions about the way that thermal operations should be modelled and understood~\cite{GibbsPres,Perry}. While the quantum system $A$ is allowed to be microscopic (say, a single particle), the heat bath $B$ is often assumed to be macroscopic, or even explicitly assumed to be unbounded. However, in many physical situations, this is not a realistic assumption. In particular, we may think of a situation where a thermodynamic protocol has to be completed in a finite amount of time, allowing $A$ to effectively interact only with a restricted number of degrees of freedom in the bath. The third law of thermodynamics~\cite{MOThirdLaw} can be seen as an example of this, as we will argue further below. This leads us to study thermal operations with bounded heat baths $B$ in this paper.

The thermodynamical questions notwithstanding, there is also a purely mathematical motivation to study finite baths, which comes from more traditional quantum information problems.
For example, consider the special case that all Hamiltonians (of system and bath) are fully degenerate, i.e.\ proportional to the identity. Then thermal operations are maps of the form
\begin{equation}
   \rho'_A = \Tr_B\left[ U\left(\rho_A\otimes \frac{\mathbf{1}_B}m\right)U^\dagger\right],
   \label{eqNoisy}
\end{equation}
with $m=\dim\, B$, and $U$ an arbitrary unitary on $AB$. These maps have been called \emph{noisy operations}~\cite{HHO,Gour} in quantum information theory. They have been studied first without explicit reference to thermodynamics, often under different names (like ``$K$-unistochastic maps''~\cite{ZB2004,Bengtsson,Musz} or ``exactly factorizable maps''~\cite{Shor,Haagerup}), formalizing the idea of a map that dilutes quantum systems with random noise.

One important question in this context is whether for two given quantum states $\rho_A$ and $\rho'_A$, there exists a noisy operation that maps $\rho_A$ to $\rho'_A$. This question has been partially answered in~\cite{HHO}: if the spectrum (vector of eigenvalues) of $\rho_A$ majorizes the spectrum of $\rho'_A$, denoted
\begin{equation}
   \rho_A\succ \rho'_A,
   \label{eqMajor}
\end{equation}
then one can obtain $\rho'_A$ to arbitrary accuracy by noisy operations from $\rho_A$. Since unitary basis changes are allowed operations, we may assume that both $\rho_A$ and $\rho'_A$ are diagonal matrices. Then~(\ref{eqMajor}) can be proven by reduction to the case of discrete classical probability distributions, where the unitary $U$ acts as a suitable permutation of basis vectors.

This strategy of ``proof by reduction to classical'' has one important drawback, however: it only works if the heat bath is assumed to be unbounded. To see this, suppose that $\rho_A$ is a pure state, and $\rho'_A$ a mixed state with irrational eigenvalues. In this case, eq.~(\ref{eqMajor}) holds. But if $U$ is a classical permutation, then the right-hand side of~(\ref{eqNoisy}) will always be a density matrix with entries that are all rational numbers, and never be equal to $\rho'_A$ which has irrational diagonal elements. In other words, to approximate $\rho'_A$ to arbitrary accuracy, the dimension $m$ of the bath has to become arbitrarily large.

It has been conjectured~\cite{Bengtsson} that one can in fact perform all transitions $\rho_A\to\rho'_A$ exactly with a finite bath of dimension $\dim B=\dim A$, by using more general unitaries $U$ which are not permutations. In this paper, we give a simple proof of this conjecture (see Lemma~\ref{LemNoisyN} below), which can be seen as a quantum version of Horn's lemma. This shows that a small bath is always sufficient to implement any noisy state transition $\rho_A\to\rho'_A$ exactly.

A ``proof by reduction to classical'' has also been the standard approach in the thermodynamic setting, where the Hamiltonians of system and bath are allowed to be non-trivial: the possibility of a transition $\rho_A\to\rho'_A$ by thermal operations has standardly been proven by resorting to the classical case, where $U$ is an (energy-preserving) permutation. This leads to a setup where one is forced to allow a heat bath $B$ which is unbounded in dimension and energy, in order to approximate all possible state transitions. As mentioned above, this assumption does not always seem to be physically realistic, in particular since it presupposes the fine control to implement arbitrary energy-preserving unitaries on the joint quantum system $AB$, where $B$ is very large.

One way that this problem has recently been addressed is by considering a small subset of ``coarse'' operations that are still sufficient to realize all possible thermal state transitions~\cite{Perry}. Similarly, Ref.~\cite{Ng} has studied the role of catalysts (auxiliary systems that ease a state transformation) that are bounded in size or energy. In this paper, we take a different route. We hold on to thermal operations, and ask whether there is an analogue of quantum Horn's lemma for them, in particular in the quasiclassical case where both initial and final state are blockdiagonal in energy.
If this was the case, then every thermal state transition between quasiclassical states could be realized exactly with a finite heat bath.

Here we give two results that answer this question in a complementary way. First, we show that ``quantum'' thermal operations (that is, ones with arbitrary energy-preserving unitaries $U$) allow indeed many transitions exactly which could only be realized approximately  by ``classical'' thermal operations (where $U$ is an energy-preserving permutation). For quantum systems with non-degenerate Hamiltonian $H_A$, Figure~\ref{fig_tqc} illustrates this: for a fixed initial state and fixed finite heat bath, the possible final states obtained by classical thermal operations form a discrete finite set. In contrast, the final states obtainable by quantum thermal operations correspond to the convex hull of this set (up to some minor amendments in the case of degeneracies). This removes the unphysical discreteness originating from the finite set of classical permutations, and paints a physically more realistic and also mathematically more convenient picture. It is formalized in our first main result, Theorem~\ref{TheQC}, which generalizes quantum Horn's lemma to thermal operations.

Figure~\ref{fig_interior} shows a simple consequence of this fact, which is the content of Corollary~\ref{CorInterior}: the set of all states that can be obtained to arbitrary accuracy from an initial state by thermal operations (that is, the states ``thermomajorized'' by the given initial state) is a compact convex set. All states in its relative interior can be obtained exactly by quantum thermal operation with a suitable finite heat bath.

However, we show that this is not in general true for the states in the relative boundary, which represents our second main result. In particular, for every non-degenerate two-level system $A$, we show explicitly that there is a state that can be obtained to arbitrary accuracy, but never exactly with any finite heat bath.

This result has an interesting physical interpretation, illustrated in Figure~\ref{fig_cooling}. Consider the special case of a two-level system with some energy gap $\Delta E$, which is initially in its (pure) excited state. Suppose we would like to couple this qubit system to a heat bath, in such a way that it is in its ground state after the interaction. Physically, this is arguably the simplest possible instance of \emph{cooling}: we draw energy from the qubit system itself (coming from the excited state) in order to cool it down to zero temperature, which is less then the ambient (heat bath) temperature. However, we show in Theorem~\ref{TheThird} that the ground state can never be obtained exactly with any finite heat bath: it can only be approximated, such that the final temperature $T'$ can never be exactly zero. We give a strictly positive lower bound on $T'$ which depends on the energy and dimension of the heat bath, which can be seen as an instance of the third law of thermodynamics.
While the third law has been examined in more detail elsewhere~\cite{MOThirdLaw}, our result suggests an intuitive picture in terms of the geometry of the state space: states of zero temperature lie on the relative boundary of the set of thermomajorized states (boundary of the light-gray region in Figure~\ref{fig_interior}), and not all of these states can be obtained with a finite heat bath.

Note that our results do not prove that zero temperature can never be attained with any finite heat bath. In the geometric picture of Figure~\ref{fig_interior}, there are states on the boundary which can be obtained exactly (for example, the initial state itself). Thus, we here only consider a very special unattainability result. On the other hand, our formulation suggests a generalization of the third law beyond cooling to zero temperature. For example, in the scenario of Figure~\ref{fig_cooling}, we can also look at initial states which are only partially excited, and obtain a potential target temperature $T'_0>0$ that can only be achieved approximately, but never exactly, by interactions with any finite heat bath.

Our paper is organized as follows. In Section~\ref{SecFramework}, we provide the mathematical framework and notation. We prove our main results for noisy operations in Section~\ref{SecNoisy}, and for thermal operations in Section~\ref{SecThermal}. In the thermal case, we first consider the case of a two-level system in Subsection~\ref{SubsecQubit}, and then the general case in Subsection~\ref{SubsecGeneral}, before discussing the physical intuition and comparison to other approaches in more detail in the Conclusions, Section~\ref{SecConclusions}.

\section{Mathematical preliminaries}
\label{SecFramework}
We will assume that all quantum systems (labelled by uppercase letters like $A$) are finite-dimensional, and carry a Hamiltonian, i.e.\ a self-adjoint operator $H_A=H_A^\dagger$. In the context of noisy operations (in particular Section~\ref{SecNoisy}), we assume that all Hamiltonians are zero. We also assume that every quantum system comes with a choice of orthonormal basis, in such a way that its Hamiltonian is diagonal in that basis.

A \emph{classical state} is a probability vector $p=(p_1\ldots,p_n)\in\R^n$, with $p_i\geq 0$, $\sum_i p_i=1$. We will treat probability vectors as column vectors in matrix calculations, but for notational simplicity, we denote them by row vectors. We identify classical states $p\in\R^n$ with the corresponding quantum states (diagonal density matrices) $\hat p$ on $\C^n$, where
\[
   p=(p_1,\ldots,p_n)\quad\Rightarrow\quad \hat p:=\left(\begin{array}{ccc}	
           p_1 & \hdots & 0 \\
           \vdots & \ddots & \vdots \\
           0 & \hdots & p_n
       \end{array}\right).
\]
Given two classical states $p,q\in\R^n$, we say that \emph{$p$ majorizes $q$}~\cite{MarshallOlkin}, and write $p\succ q$, if and only if
\[
   \sum_{i=1}^k p_i^\downarrow \geq \sum_{i=1}^k q_i^\downarrow\qquad \mbox{for all }k=1,\ldots,n,
\]
where $p^\downarrow = (p_1^\downarrow,\ldots,p_n^\downarrow)$ denotes the entries of $p$ in non-increasing order, i.e\ $p_1^\downarrow\geq p_2^\downarrow\geq \ldots\geq p_n^\downarrow$. Furthermore, we write $\rho\succ\sigma$ for quantum states $\rho,\sigma$ if and only if ${\rm spec}(\rho)\succ{\rm spec}(\sigma)$, where ${\rm spec}(\rho)$ denotes the vector of eigenvalues of $\rho$. According to the Hardy-Littlewood-Pol\'ya theorem~\cite{MarshallOlkin}, $p\succ q$ if and only if there exists a bistochastic matrix $D\in\R^{n\times n}$, i.e.\ a matrix with $D_{ij}\geq 0$, $\sum_i D_{ij}=\sum_j D_{ij}=1$, such that $q=Dp$.
The following lemma is easily checked by direct calculation, but will be very useful in the following. It has also been observed in~\cite{RuchMead}.
\begin{lemma}[Ruch and Mead~\cite{RuchMead}]
\label{LemUnistoch}
	Let $p\in\R^n$ be a classical state, and let $U\in{\rm U}(n)$ be a unitary. Then the vector of diagonal elements of $U\hat p U^\dagger$ can be written
	\[
	\left( U\hat p U^\dagger)_{11},\ldots,(U\hat p U^\dagger)_{nn}\right)= Dp,
	\]
	where $D$ is a bistochastic matrix with entries $D_{ij}=|U_{ij}|^2$, i.e.\ $D=U\circ U^*$ for the Hadamard~\cite{Bhatia} product $\circ$, i.e.\ $(X\circ Y)_{ij}=X_{ij}Y_{ij}$.
\end{lemma}
Given any quantum system $A$, a \emph{noisy operation} is a map on the quantum states of $A$ which can be written in the form~(\ref{eqNoisy}), for some suitable choice of quantum system $B$ of dimension $\dim B=m$ and unitary transformation $U$ on $AB$. (We will not consider more general noisy or thermal operations which map states on one Hilbert space to states on another Hilbert space in this paper.) The known relation between majorization and noisy operations is as follows.
\begin{lemma}[\cite{HHO,Bengtsson}]
If $\rho\to\rho'$ by some noisy operation, then $\rho\succ\rho'$. Conversely, if $\rho\succ\rho'$, then for every $\varepsilon>0$, there exists some noisy operation $\Phi$ with $\Phi(\rho)=\rho'_\varepsilon$, where $\|\rho'_\varepsilon-\rho'\|<\varepsilon$. If both $\rho$ and $\rho'$ are diagonal in the canonical eigenbasis, then $\Phi$ can be (and is standardly) chosen in such a way that the unitary $U$ in~(\ref{eqNoisy}) is a (``classical'') permutation of basis vectors.
\end{lemma}
In the case of non-trivial Hamiltonians, all these notions are generalized. We always assume that there is a fixed background inverse temperature $\beta=1/(k_B T)$, with $k_B$ the Boltzmann constant. The classical Gibbs state of a system $B$ with Hamiltonian $H_B$ is $\gamma_B=(e^{-\beta E_1},\ldots,e^{-\beta E_n})/Z$, where $E_1,\ldots,E_n$ are the eigenvalues (and thus diagonal elements) of $H_B$, and $Z=\sum_{i=1}^n e^{-\beta E_i}$ is the partition function. The corresponding quantum Gibbs state is $\hat\gamma_B=\exp(-\beta H_B)/Z$. We will sometimes write $\gamma^B$ for $\gamma_B$ if we need space for indices.

A \emph{thermal operation} on a quantum system $A$ is a map of the form
\begin{equation}
   \rho'_A = \Tr_B\left[
      U(\rho_A\otimes\hat\gamma_B)U^\dagger
   \right],
   \label{eqThermalMain}
\end{equation}
where $U$ is an energy-preserving unitary, i.e.\ $[U,H_A+H_B]=0$ (this constraint becomes void in the case that $H_A=H_B=0$, which recovers the noisy operations). Any choice of bath $B$ is allowed, but the bath must be in its thermal state $\hat\gamma_B$. An important property of noisy operations $\Phi$ is that they cannot generate coherences between different energies: that is, if $[\rho,H_A]=0$, then $[\Phi(\rho),H_A]=0$. That is, blockdiagonal states (in the energy eigenbasis) are mapped to blockdiagonal states.

Thermal operations are closely related to a generalized notion of majorization, named \emph{thermomajorization} (this also corresponds to ``majorization relative to $d$'' for $d=\gamma$ in~\cite{MarshallOlkin}). For any two classical states $p,p'$ on a system $A$ with Gibbs state $\gamma=\gamma_A$, we write
\[
   p\succ_\gamma p'
\]
and say that ``$p$ thermomajorizes $p'$'' if and only if there exists a stochastic $\gamma$-preserving matrix $D$ such that $p'=Dp$. Stochasticity of $D$ means that $D$ maps states to states, i.e.\ $D_{ij}\geq 0$ and $\sum_i D_{ij}=1$ for all $j$, and $\gamma$-preservation is the property that $D\gamma=\gamma$.

The following was shown (in different notation) by Janzing et al.~\cite{Janzing}, see also~\cite{Ruch1978,Ruch1980}.
\begin{lemma}[Janzing et al.~\cite{Janzing}]
\label{LemJanzing}
If $\hat p\to\hat p'$ by some thermal operation, where $p$ and $p'$ are classical states on the same system $A$, then $p\succ_\gamma p'$, where $\gamma=\gamma_A$ is the Gibbs state of $A$.

Conversely, if $p\succ_\gamma p'$, then for every $\varepsilon>0$, there exists some thermal operation $\Phi$ with $\Phi(\hat p)=\hat p'_\varepsilon$, where $\|p'_\varepsilon-p'\|<\varepsilon$. This map $\Phi$ can be chosen in such a way that the unitary $U$ in~(\ref{eqNoisy}) is a permutation of basis vectors; the bath Hamiltonian $H_B$ can be chosen to be composed of a sufficient number of non-interacting copies of systems with Hamiltonian $H_A$.
\end{lemma}
In the following, we are only interested in transitions between classical states, i.e.\ instances of~(\ref{eqThermalMain}) where input and output state are both diagonal in the canonical (energy eigen)basis. Note that this does not mean that we restrict to operations $\Phi$ that map \emph{all} diagonal input states to diagonal output states.

We are thus not considering the role of quantum coherence in this paper, which is currently an active area of research~\cite{LJR2015,Cwiklinski,KLOJ2016}. However, we are going beyond the earlier quasiclassical setups by explicitly considering unitaries $U$ in~(\ref{eqThermalMain}) which are not just classical permutations of basis vectors.

\section{Noisy operations with small auxiliary systems}
\label{SecNoisy}
As mentioned in the introduction, in most of the literature on single-shot quantum thermodynamics (e.g.~\cite{HHO,HO2013,Resource2013,Brandao2015,Gour}), noisy operations are constructed by implementing a classical permutation in a high-dimensional degenerate eigenspace of system and bath. To approximate target states to arbitrary accuracy, this needs in general arbitrarily large auxiliary systems. (We also use the word ``bath'' for these auxiliary systems for brevity, even though they do not carry a non-trivial Hamiltonian and are thus not heat baths in a strict sense.)

We will now show that this is not necessary if we consider more general unitaries. In this case, all transitions can be implemented exactly with a bath that is of the same dimension as the system.
\begin{lemma}
\label{LemNoisyN}
If $p,p'\in\R^n$ are states such that $p\succ p'$, then the transition $p\to p'$ can be accomplished by a noisy operation with an auxiliary system of size $n$. That is, there is a unitary $U$ on $AB$, $A=B=\mathbb{C}^n$, such that
\begin{equation}
   \hat p'=\Tr_B\left[U\left(\hat p\otimes\frac {\mathbf{1}} n \right)U^\dagger\right].
   \label{eqPPPrime}
\end{equation}
\end{lemma}
\begin{proof}
We will apply the \emph{Schur-Horn Theorem}~\cite{Horn}, which tells us that there is a unitary $V$ on $A$ such that $V\hat p V^\dagger$ has diagonal elements $p'_1,\ldots,p'_n$. Therefore, $\hat p'=\Phi(V\hat p V^\dagger)$, where $\Phi$ is the ``decoherence map'', satisfying
\[
   \langle i|\Phi(\rho)|j\rangle=\left\{
      \begin{array}{cl}
         \langle i|\rho|j\rangle & \mbox{if }i=j\\
         0 & \mbox{otherwise}.	
      \end{array}
   \right.
\]
It remains to show that $\Phi$ can be implemented as a noisy operation with a bath of size $n$. To this end, choose an arbitrary orthonormal basis $|1\rangle,\ldots,|n\rangle$ on $B=\C^n$, and let $\pi$ be the cyclic permutation with $\pi|i\rangle=|i+1\rangle$ if $i<n$ and $\pi|n\rangle=|1\rangle$. Define the unitary $W$ by
\[
   W=\sum_{k=1}^n |k\rangle\langle k|\otimes \pi^k.
\]
Then, if $\rho'=\Tr_B\left[ W\left(\rho\otimes\frac{\mathbf{1}} n \right) W^\dagger\right]$, we have $\langle i|\rho'|j\rangle=\frac 1 n \langle i|\rho|j\rangle \tr(\pi^{i-j})$, which is zero if $i\neq j$ and $\langle i|\rho|j\rangle$ otherwise. In other words, the unitary $W$ implements $\Phi$ as a noisy operation, and we can set $U:=W(V\otimes\mathbf{1}_B)$.
\end{proof}
This result confirms a conjecture by Bengtsson and \.{Z}yczkowski~\cite{Bengtsson}, who however use slightly different terminology. They call a noisy operation on a system of size $n$, mediated by a bath of size $n$, ``unistochastic'', regarded as a map between quantum states (note that this conflicts with our use of this term in Lemma~\ref{Lem3} in the appendix). For two quantum states $\rho$ and $\sigma$, write $\rho\succ\sigma$ if and only if the spectrum of $\rho$ majorizes the spectrum of $\sigma$. Then, in the terminology of~\cite{ZB2004,Bengtsson,Musz}), Lemma~\ref{LemNoisyN} implies ``Quantum Horn's Lemma'': \emph{If $\rho\succ\sigma$, then there exists a unistochastic quantum operation that maps $\rho$ to $\sigma$ exactly.} Clearly, if $p$ and $p'$ denote the probability vectors of eigenvalues of $\rho$ and $\sigma$, then there is a unitary $U$ such that~(\ref{eqPPPrime}) holds true, with a bath of size $n$. Since there are unitaries $V$ and $W$ such that $\hat p=V\rho V^\dagger$ and $\hat p'=W\sigma W^\dagger$, we can set $U':=(W^\dagger\otimes\mathbf{1})U(V\otimes\mathbf{1})$, and then $U'$ will generate a ``unistochastic quantum operation'' that maps $\rho$\ to $\sigma$.

It is interesting to see that Lemma~\ref{LemNoisyN} above is closely related to (and, in fact, has an alternative proof in terms of) a variant of the quantum marginal problem as formulated, for example, in~\cite{Daftuar}: given the spectrum of some bipartite quantum state $\rho_{AB}$, then what are the possible spectra of the local reduced state $\rho_A$?

The following lemma gives a complete answer to this question in the case $\dim A\leq \dim B$. In a nutshell, it says that the possible spectra of $\rho_A$ are those that are majorized by the partial trace of the (non-increasingly ordered) spectrum of $\rho_{AB}$.

\begin{lemma}
\label{LemMarginal}
Let $\rho$ be a quantum state on $AB$, and $\sigma$ a quantum state on $A$. If $\dim A\leq \dim B$, there exists a unitary $U$ on $AB$ with
\[
   \Tr_B\left(U\rho U^\dagger\right)=\sigma
\]
if and only if
\begin{equation}
   \Tr_B \hat\lambda(\rho)\succ\sigma,
   \label{eqMarginal}
\end{equation}
where $\lambda(\rho)$ is the vector of eigenvalues of $\rho$, ordered in non-increasing order, and $\hat\lambda(\rho)$ is the corresponding diagonal density matrix.
\end{lemma}
This lemma was conjectured in~\cite{Daftuar} in 2005, where it was also shown that~(\ref{eqMarginal}) is \emph{necessary} for the existence of a suitable unitary $U$ which satisfies the statement of the lemma. After Lemma~\ref{LemMarginal} was proven by the first author in his Master thesis~\cite{Scharlau}, it turned out (cf.\ Acknowledgments) that the statement had just recently been proven independently in~\cite{LiPoonWang}. The lemma can also be inferred from the results in~\cite{Vergne}. We give our proof in the appendix for completeness.

Lemma~\ref{LemNoisyN} follows as a corollary from Lemma~\ref{LemMarginal} by setting $\dim A=\dim B=n$, $\rho=\rho_A\otimes \mathbf{1}/n$, $\sigma=\rho'_A$, and noting that
\[
   \lambda(\rho)=\frac 1 n(\underbrace{\lambda_1,\ldots,\lambda_1}_n,\underbrace{\lambda_2,\ldots,\lambda_2}_n,\ldots,\underbrace{\lambda_n,\ldots,\lambda_n}_n),
\]
where $\lambda_1\geq\lambda_2\geq\ldots$ are the eigenvalues of $\rho_A$. The partial trace corresponds to summing over $n$-blocks, thus $\Tr_B \hat\lambda(\rho)=\hat\lambda$, where $\lambda=(\lambda_1,\ldots,\lambda_n)$. These are the eigenvalues of $\rho_A$, which must therefore majorize the eigenvalues of $\rho'_A$.

Given that all possible transitions on a system $A$ of size $n$ can be achieved with a bath $B$ of size $n$, it is natural to ask whether one can more generally classify the set of all \emph{maps} on the classical probability distributions that can be implemented in this setting. We give a partial answer in Lemma~\ref{Lem3} in the appendix, where we show that the corresponding set of maps lies somewhere in between the unistochastic and the bistochastic maps.

The relation to the quantum marginal problem allows us to explore what happens in the case that the heat bath dimension $m=\dim B$ is smaller than the system dimension $n=\dim A$. The case $m=2, n=3$ has been considered in~\cite{Daftuar} in Sec.\ 7.3, where the authors have worked out the full set of inequalities relating the eigenvalues of $\rho_{AB}$ to the eigenvalues of possible marginals $\rho_A$. Their result implies the following:
\begin{lemma}[Case $n=3$, $m=2$]
\label{Lem23}
If $p,p'\in\R^3$ are states such that $p\succ p'$, then the transition $p\to p'$ can be accomplished by a noisy operation with an auxiliary system of size $m=2$. That is, there is a unitary $U$ on $AB$, $A=\C^3$, $B=\mathbb{C}^2$, such that
\[
   \hat p'=\Tr_B\left[U\left(\hat p\otimes\frac {\mathbf{1}} 2 \right)U^\dagger\right].
\]
\end{lemma}
That is, in this case, even a bath that is smaller than the system is sufficient to implement all noisy state transitions.

We do not currently know the smallest possible value of $m$ that would allow all noisy state transitions on a system of dimension $n$. However, we have the following elementary bound:
\begin{lemma}
\label{LemSqrt}
If all noisy state transitions on a system of size $n$ can be implemented with a bath of size $m$, then $m\geq\sqrt{n}$.
\end{lemma}
\begin{proof}
If all noisy state transitions can be implemented, then in particular a pure state can be mapped to a maximally mixed state. That is, there is a pure state $|\psi\rangle$ on $A=\C^n$ and a unitary $U$ on $AB$ such that
\begin{eqnarray*}
\frac{\mathbf{1}} n &=&\Tr_B\left[ U \left(|\psi\rangle\langle\psi|_A\otimes \frac{\mathbf{1}} m \right)U^\dagger\right]\\
&=& \frac 1 m \sum_{i=1}^m \Tr_B\left[ U\left(|\psi\rangle\langle\psi|_A\otimes |i\rangle\langle i|_B\right)U^\dagger\right],
\end{eqnarray*}
where $\{|i\rangle_{i=1}^m$ is any orthonormal basis of $B=\C^m$. The Schmidt rank of the pure state $U|\psi\rangle_A|i\rangle_B$ is upper-bounded by $m$, and so the rank of every density operator $\Tr_B[\ldots]$ on the right-hand side is upper-bounded by $m$. Due to the subadditivity of the rank, this means that the rank of the right-hand side is less than or equal to $m^2$. Hence ${\rm rank}(\mathbf{1}/n)=n\leq m^2$.
\end{proof}
We do not know whether this bound is tight. If it was, then this would say that any noisy operation on $k$ qubits can be achieved with a bath of $k/2$ qubits only (rounding up accordingly). This would be a significant tightening of quantum Horn's lemma.

\section{Thermal operations with small heat baths}
\label{SecThermal}
We now turn to the more general case of thermal operations. We will start by analyzing two-level systems in Subsection~\ref{SubsecQubit}, which illustrates most of the features of arbitrary-dimensional systems that will be proven in Subsection~\ref{SubsecGeneral}.
\subsection{The qubit case with a finite heat bath}
\label{SubsecQubit}
Consider a two-dimensional Hilbert space with orthonormal basis $\{\ket{1},\ket{2}\}$ and Hamiltonian
\[
	H_A = \Delta E \ \ketbra{2}{2},
\]
and assume that $\Delta E > 0$. The corresponding Gibbs state at inverse temperature $\beta>0$ is given by 
\[
{\hat \gamma}_A =\bigl(\,\ketbra{1}{1}+ e^{-\beta \Delta E}\ketbra{2}{2}\,\bigr) / \bigl(\, 1 + e^{-\beta \Delta E}\,\bigr).
\]
Now we would like to see which (classical) states $p'$ are thermomajorized by a given state $p$. According to Lemma~\ref{LemJanzing}, we need a stochastic matrix $D$ which satisfies
\[
	Dp = p' \quad\text{ and }\quad D\gamma = \gamma.
\]
In two dimensions the stochastic matrices that preserve $\gamma$ are parametrized by a single parameter $\alpha$:
\begin{equation}
D_\alpha= \begin{pmatrix}
\ 1-\alpha & \alpha\, {\gamma_1}/{\gamma_2} \\[.3em] \alpha & 1- \alpha \, \gamma_1 / \gamma_2 \
\end{pmatrix}
\label{eqGeneralD}
\end{equation}
with the requirement $0\leq\alpha \leq \gamma_2/\gamma_1$, where $\gamma_1,\gamma_2$ are the components of $\gamma=\gamma_A$.  The set of states thermomajorized by $p$ is thus also parametrized by $\alpha$, and we can use the notation $p'_\alpha:=D_\alpha p$. Three special cases are as follows:
\begin{itemize}
	\item $\alpha=0$ gives $D_\alpha=\mathbf{1}$ and $p'_\alpha=p$.
	\item $\alpha=\gamma_2$ gives $D_\alpha=\left(\begin{array}{cc} \gamma_1 & \gamma_1 \\ \gamma_2 & \gamma_2\end{array}\right)$, which maps all states $p$ to the thermal state $p'_\alpha=\gamma$.
	\item $\alpha=\gamma_2/\gamma_1=e^{-\beta\Delta E}$ gives the final state
	\begin{equation}
	    p'_\alpha=p^*:=\left(1-\frac{\gamma_2}{\gamma_1}p_1, \frac {\gamma_2}{\gamma_1} p_1\right).
	    \label{eqPStar}
	    \end{equation}
\end{itemize}
The set of states thermomajorized by $p$, including these special cases, is shown in Figure~\ref{fig_qubit}. In the case $\Delta E = 0$, the $D_\alpha$ are bistochastic, and we recover the case of standard majorization from Section~\ref{SecNoisy}.

\begin{figure}[!hbt]
\begin{center}
\includegraphics[angle=0, width=.4\textwidth]{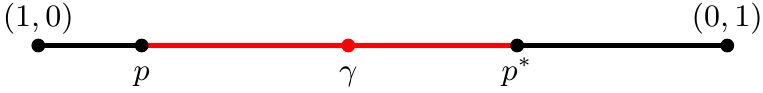}
\caption{The line shows the diagonal states of a qubit system with $(1,0)$ and $(0,1)$ corresponding to the eigenstates $\ket{1}$ and $\ket{2}$. The states thermomajorized by $p$ correspond to the points between $p$ and $p^*$.}
\label{fig_qubit}
\end{center}
\end{figure}
In order to implement the $D_\alpha$ as thermal operations, it is necessary to choose $H_B$ in a way that $H_A+H_B$ exhibits some degeneracies, since the unitaries in~(\ref{eqThermalMain}) have to commute with the total Hamiltonian. So the spectrum of $H_B$ has to contain the energy gap $\Delta E$, and furthermore it seems plausible that carrying the energy gap between many different energy levels should allow for a greater variety of transitions.

Let us therefore consider a truncated harmonic oscillator
\begin{equation}
H_B =  \Delta E \sum_{i=1}^{m} (i-1) \ketbra{i}{i},
\label{eqOszi}
\end{equation}
that is, an $m$-level system that has $\Delta E$ as energy difference between all neighbouring eigenstates. In the case $m=2$, we have one copy of the original system, and for $m\to\infty$ we obtain a harmonic oscillator.
The joint Hamiltonian of system and bath $H_{AB} = H_A + H_B$ becomes
\begin{equation*}
H_{AB}	= \Delta E\sum_{i=1}^2 \sum_{j=1}^m (i+j-2)\ketbra{i,j}{i,j}
\end{equation*}
where $\ket{i,j}_{AB} = \ket i_A \otimes \ket j_B$.  All energy levels of $H_{AB}$ except for the lowest and highest one are doubly degenerate. As explained above, this degeneracy is necessary to implement non-trivial thermal operations.

In the case $m=2$, we can implement one of the special cases above (apart from the trivial case), namely $D_\alpha$ for $\alpha=\gamma_2$. This is achieved by the energy-preserving permutation $\pi$ with $\pi\ket{12} = \ket{21}$ and $\pi\ket{21}=\ket{12}$ which swaps the two systems, leading to
$ \hat p \otimes \hat \gamma \mapsto 	\hat \gamma \otimes \hat p$.
Tracing out the second system, we end up with the Gibbs state of the qubit, independent of the state we started with.

By implementing other unitaries $U_{2\times 2}$ on the subspace ${\rm span}\{|12\rangle,|21\rangle\}$, we can in fact implement all maps $D_\alpha$ with $0\leq\alpha\leq \gamma_2$, and thus obtain all corresponding target state $p'_\alpha$. This follows from the simple observation that ${\rm SU}(2)$ is connected, and thus by varying over all $U_{2\times 2}\in {\rm SU}(2)$, we obtain target states $p'$ that vary continuously between $p'=p$ and $p'=\gamma$.

Can we obtain more target states than this? To answer this in the negative, consider the parametrization
\[
   U_{2\times 2}=\left(
      \begin{array}{cc}
      	  \kappa & -\lambda^* \\ \lambda & \kappa^*
      \end{array}
   \right)
\]
with $\kappa,\lambda\in\C$ such that $|\kappa|^2+|\lambda|^2=1$. Implementing this unitary on the degenerate subspace, we obtain a target state $p'(U_{2\times 2})$. A simple calculation shows that
\[
   p'(U_{2\times 2})=|\kappa|^2 p +|\lambda|^2 \gamma=
   |\kappa|^2 p'(\mathbf{1}_{2\times 2})+|\lambda|^2 p'(\pi_{2\times 2}),
\]
where $\pi$ is the swap introduced above. This shows that the extremal cases are generated by classical permutations (swap or identity on the degenerate subspace) and choosing general energy-preserving unitaries give us all states that are in the convex hull or ‘in between’ these two states, but not more.

Now we would like to see which states we can achieve for $m>2$. Consider the permutation
\begin{equation*}
\pi_\text{max}: \quad \ket{1,j} \longleftrightarrow \ket{2,j-1} \qquad \text{for } j=2,\ldots,m	
\end{equation*}
which swaps the basis vectors of all two-dimensional energy eigenspaces of $H_{AB}$. 
Intuitively, $\pi_{\max}$ is the classical permutation for which we expect to obtain the $D_\alpha$ with the largest possible value of $\alpha $ (for this specific heat bath). Calculating the thermal operation for $p=(1,0)$ and reading off $1-\alpha$ from $p'=(1-\alpha,\alpha)$ shows that $\pi_{\max}$ yields
\begin{equation*}
	1-\alpha = \gamma_1^B = \frac 1{\tr(e^{-\beta H_B})}.
\end{equation*}
which is indeed the minimal value for $1-\alpha$ that is achievable by choosing any energy-preserving permutation. This becomes clear from the the fact that all energy-preserving permutations can be written as $\pi_J$, where $J \subset \{2,\ldots,m\}$, and
\[
   \pi_J|1,j\rangle :=\left\{
      \begin{array}{ll}
      	  |2,j-1\rangle & \mbox{if }j\in J, \\
      	  |1,j\rangle & \mbox{otherwise.}
      \end{array}
    \right.
\]
That is, $J$ indicates in which subspaces we choose to swap the two basis vectors, and $\pi_{\max}=\pi_{\{2,\ldots,m\}}$. Choosing some permutation $\pi_J$ leads a corresponding value $\alpha$ of $1-\alpha = \sum_{j\in \{1,\ldots,m\}\setminus J} \gamma_j^B$ which is indeed minimal for $J=\{2,\ldots,m\}$.
We thus define 
\begin{equation}
\alpha_{\text{max}}^{(m)}= 1- \gamma_1^B 
	= 1- \frac{1-e^{-\beta \Delta E}}{1-e^{-m\,\beta\Delta E}},
\end{equation}
which is the maximal value of the parameter $\alpha$ that we can achieve with this kind of heat bath of dimension $m$, even if we choose not just a permutation but a more general unitary for the thermal operation. 
This can be seen similarly as in the case $m=2$. The energy spectrum of $H_{AB}$ is $0,\Delta E,\ldots,m\Delta E$, and all energies except for the highest and lowest are two-fold degenerate. Thus, all energy-preserving unitaries consist of $(m-1)$ independently chosen unitaries from ${\rm SU}(2)$ acting in the respective energy eigenspaces. For every single one of these eigenspaces, choosing a general unitary instead of a swap or the identity (as in $\pi_J$) will give a state that is a mixture of the two permutation cases, as above. Thus we can conclude that \emph{all target states $p'$ obtainable by energy preserving unitaries are in the convex hull of those that we obtain by energy preserving classical permutations.} 

So the set of states $p'$ that we can obtain with a truncated harmonic oscillator as a heat bath is
\[
   {\rm conv}\left\{p,p'_{\alpha_{\max}^{(m)}}\right\}.
\]
This means that the set of achievable states $p'$ is growing with increasing $m$, but never exhausts the set of all thermomajorized states: $p^*$, corresponding to $\alpha=e^{-\beta \Delta E}$, can never be achieved exactly for finite $m$. However, we can get arbitrarily close to $p^*$ by increasing the dimension $m$ since $\lim_{m \rightarrow \infty}\alpha_{\max}^{(m)} = e^{-\beta \Delta E}$. In particular, all states $p'\neq p^*$ which are thermomajorized by $p$ can be obtained \emph{exactly} with a finite heat bath of suitable size. This is also sketched in Figure~\ref{fig_oscillator}.

\begin{figure}[!hbt]
\begin{center}
\includegraphics[angle=0, width=.4\textwidth]{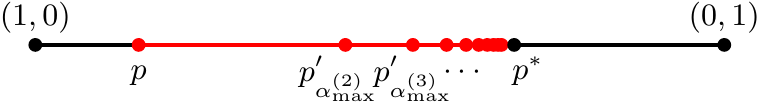}
\caption{The states $p'_{\alpha_{\max}^{(m)}}$ for $m=2,\ldots, 10$. They converge to $p^*$ in the limit $m\to\infty$, but they never reach it. That is, no finitely truncated harmonic oscillator can allow all possible state transitions on the qubit.}
\label{fig_oscillator}
\end{center}
\end{figure}

What if we choose any other (finite) heat bath Hamiltonian $H_B$ which is \emph{not} of the truncated harmonic oscillator form -- could it be that we achieve the state $p^*$ exactly? The following result answers this question in the negative, and gives a bound on how close we can get to $p^*$. It also says that truncated harmonic oscillators are in a certain sense optimal.

\begin{lemma}
\label{LemHarmonicOsci}
Let $A$ be a qubit system with energy gap $\Delta E>0$, and $B$ a finite-dimensional heat bath with arbitrary Hamiltonian $H_B$. Then the stochastic maps $D_\alpha$ of the form~(\ref{eqGeneralD}) which can be implemented by thermal operations on this system and heat bath satisfy
\begin{equation}
   \alpha\leq e^{-\beta\Delta E}\left(1-g_{\max}^B \gamma_{\max}^B\right),
   \label{eqdIneq}
\end{equation}
where $\gamma_{\max}^B=e^{-\beta E_{\max}^B}/\tr(e^{-\beta H_B})$ is the thermal occupation of the highest energy level of $H_B$, and $g_{\max}^B$ is its degeneracy. Moreover, inequality~(\ref{eqdIneq}) is tight, i.e.\ there is some thermal operation that achieves this value of $\alpha$, if and only if
\begin{itemize}
	\item for every energy eigenvalue $E$ of $H_B$, except for the highest, also $E+\Delta E$ is an energy eigenvalue, and
	\item the degeneracies satisfy $g_E^B\geq g_{E+\Delta E}^B$.
\end{itemize}
In particular, given some classical initial state $p$ on $A$, the state $p^*$ (corresponding to $\alpha=e^{-\beta\Delta E}$ as in~(\ref{eqPStar})) can never be created exactly from $p$ with any finite heat bath. However, all other states that are thermomajorized by $p$ can be obtained exactly with a finite heat bath.
\end{lemma}
\begin{proof}
Denote the orthonormal basis vectors of $B$ by $|i\rangle$ and the corresponding energies by $E_i$ (there may be $i\neq j$ with $E_i=E_j$). Furthermore, denote the set of ``matching'' energies of $H_B$ by $M_B$, where an energy $E$ is matching if there is some $i$ such that $E_i=E$, \emph{and} if there is some $j$ such that $E_j=E-\Delta E$. That is, the matching energies $E\in M_B$ are those energy values which have a suitable energy gap downwards in order to allow for thermal operation in conjunction with the qubit on $A$.

To every $E\in M_B$, there is a corresponding global energy eigenspace
\[
   S_E:={\rm span}\{|1i\rangle,|2j\rangle\,\,|\,\, E_i=E\mbox{ and } E_j=E-\Delta E\},
\]
and clearly $E,E'\in M_B$ with $E\neq E'$ implies $S_E \perp S_{E'}$. Let $P_E$ be the orthogonal projector onto $S_E$, and $P_c:=\mathbf{1}-\sum_{E\in M_B} P_E$ the projector onto the orthogonal complement $S_c$ of the span of all these (matching) energy eigenspaces.

Any energy-preserving unitary must be a direct sum of unitaries $U_E$ acting on $S_E$ and $U_c$ acting on $S_c$. Define
\[
   (\hat p\otimes \hat\gamma)_E:= P_E (\hat p\otimes \hat \gamma)P_E
\]
and similarly $(\hat p\otimes\hat \gamma)_c$. Any thermal operation $D_\alpha$ acting on a probability vector $p$ which can be implemented with the given heat bath can thus be written
{\small
\[
   \left(D_\alpha(p)\right)^{\wedge}=
   \Tr_B\left(\sum_{E\in M_B} U_E (\hat p\otimes\hat\gamma)_EU_E^\dagger + U_c(\hat p\otimes\hat\gamma)_cU_c^\dagger\right).
\]}
In particular, since thermal operations do not generate off-diagonal elements~\cite{HO2013}, the right-hand side will yield a diagonal density matrix on $A$. Due to~(\ref{eqGeneralD}), we can determine the corresponding parameter $\alpha$ by setting $p=(1,0)$ and computing $\alpha=\langle 2|\left(D_\alpha(p)\right)^\wedge |2\rangle$. Since
\[
   (\hat p\otimes \hat \gamma)_c=|1\rangle\langle 1|\otimes \sum_{i:\,\, E_i\not\in M_B} \gamma_i^B |i\rangle\langle i|
\]
with $\gamma_i^B:=\exp(-\beta E_i)/Z$ and $Z=\tr\left(e^{-\beta H_B}\right)$, the energy-preserving unitary $U_c$ cannot map any $|1i\rangle$ in this sum to any state of non-zero overlap with $|2j\rangle$ for any $j$. Thus, $\Tr_B U_c(\hat p \otimes \hat \gamma)_c U_c^\dagger=|1\rangle\langle 1|$, and so
\begin{eqnarray*}
	\alpha &=& \langle 2|\left( \Tr_B \sum_{E\in M_B} U_E (\hat p\otimes \hat \gamma)_E U_E^\dagger\right) |2\rangle \\
	&=&\sum_{E\in M_B} \tr\left( |2\rangle\langle 2|\otimes\mathbf{1}_B\, P_E U_E(\hat p\otimes\hat\gamma)_E U_E^\dagger P_E\right)\\
	&\stackrel{(*)}\leq& \sum_{E\in M_B} \left\| P_E |2\rangle\langle 2|\otimes \mathbf{1}_B P_E\right\|_1 \cdot\|(\hat p\otimes \hat\gamma)_E\|_\infty\\
	&=&\sum_{E\in M_B} g_{E-\Delta E}^B e^{-\beta E}/Z,
\end{eqnarray*}
where $g_{E'}^B$ is the degeneracy of the energy $E'$ in $H_B$. Thus
\begin{eqnarray*}
	\alpha&\leq& \frac 1 Z \sum_{j:\,\, E_j+\Delta E\in M_B} e^{-\beta (E_j+\Delta E)}\\
       &=&\frac 1 Z e^{-\beta\Delta E}\sum_{j:\,\, E_j+\Delta E\in M_B}e^{-\beta E_j}\\
       &\stackrel{(**)}\leq & \frac 1 Z e^{-\beta \Delta E}\sum_{j:\,\, E_j\neq E_{\max}^B} e^{-\beta E_j}\\
       &=& \frac 1 Z e^{-\beta \Delta E}\left(Z-g_{\max}^B\cdot e^{-\beta E_{\max}^B}\right), 
\end{eqnarray*}
where $g_{\max}^B\in\mathbb{N}$ is the degeneracy of the highest energy level of $H_B$. Inequality $(**)$ holds with equality if and only if for all eigenvalues $E\neq E_{\max}^B$ of $H_B$ also $E+\Delta E$ is an eigenvalue. Furthermore, $(*)$ becomes an equality for some suitable choice of unitary if and only if every global eigenstate $|1i\rangle$ with $E_i=E$ can be unitarily mapped to another global energy eigenstate $|2j\rangle$ with $E_j=E-\Delta E$, which is possible by some suitable $U_E$ if there are enough of the latter.
\end{proof}	
Lemma~\ref{LemHarmonicOsci} has a compelling physical interpretation in terms of the \emph{third law of thermodynamics}. Think of the task to cool a physical system. In general, this needs two ingredients: first, a heat bath, and second, a work reservoir. We can model both by considering a qubit system in a  state $p=(0,1)$, that is, an excited eigenstate of energy $\Delta E$. Since $p$ thermomajorizes all other classical qubit states $p'$, it is clear that this state contains enough free energy to allow a transition to any other state $p'$ to arbitrary accuracy by coupling it to a suitable heat bath.

Suppose we would like to obtain another state
\begin{equation}
   p'=(1,e^{-\beta'\Delta E})/(1+e^{-\beta' \Delta E})
   \label{eqPStrich}
\end{equation}
of inverse temperature $\beta'=1/(k_B T')$. The third law of thermodynamics tells us that we cannot achieve temperature $T'=0$ in finite time. More in detail, Nernst's formulation~\cite{MOThirdLaw}, first stated in 1912~\cite{Nernst}, can be phrased as follows~\cite{Loebl,FowlerGuggenheim}: \emph{It is impossible by any procedure, no matter how idealized, to reduce any assembly to the absolute zero in a finite number of operations.} The thermal state of temperature zero is the ground state $p'=(1,0)$ which is nothing but the state $p^*$ from~(\ref{eqPStar}) (for our choice of initial state $p$). And, in fact, Lemma~\ref{LemHarmonicOsci} says that we cannot obtain $p'=p^*$ exactly with any finite heat bath. But we can say more.
\begin{theorem}[Third law of thermodynamics for qubits]
\label{TheThird}
   Let $A$ be a qubit system with energy gap $\Delta E$ which is in its excited energy eigenstate, and $B$ an $n$-dimensional heat bath with smallest and largest energy eigenvalues $E_{\min}^B$ and $E_{\max}^B$, which is in a thermal state at temperature $T$. If we apply any thermal operation that yields a new qubit state $p'$ with non-negative temperature $T'\geq 0$ (i.e.\ without population inversion), then
   \[
      T'\geq \frac{T\Delta E}{E_{\max}^B-F_B}
      \geq \frac{T\Delta E}{E_{\max}^B-E_{\min}^B+k_B T \log n},
   \]
   where $F_B=-k_B T \log \tr(\exp(-H_B/(k_B T)))$ is the free energy of the heat bath.
\end{theorem}
\begin{proof}
If we apply a map $D_\alpha$ of the form~(\ref{eqGeneralD}) to the excited state $p=(0,1)$, then $D_\alpha p=p'$ with $p'$ as in~(\ref{eqPStrich}) yields
\[
   \beta'=-\frac 1 {\Delta E}\log\left(\frac{e^{-\beta\Delta E}}{\alpha}-1\right),
\]
and the argument inside the $\log$ is strictly positive due to $\alpha<e^{-\beta\Delta E}$. Furthermore, $\beta'\in [0,+\infty]$ implies that $\alpha\geq \frac 1 2 e^{-\beta\Delta E}$. From eq.~(\ref{eqdIneq}), it follows in particular that $\alpha\leq e^{-\beta\Delta E}\left(1-\gamma_{\max}^B\right)$ (since $g_{\max}^B\geq 1$), and so
\begin{eqnarray*}
   \beta'&\leq&-\frac 1 {\Delta E}\log\left[ \left(1-\gamma_{\max}^B\right)^{-1}-1\right]\\
   &=& \frac 1 {\Delta E}\log\left( \left(\gamma_{\max}^B\right)^{-1}-1\right)
   \leq -\frac 1 {\Delta E}\log\gamma_{\max}^B\\
   &=& \frac 1 {\Delta E}\left(\beta E_{\max}^B+\log Z_B\right),
\end{eqnarray*}
where $Z_B=\tr(\exp(-\beta H_B))$, with $\beta=1/(k_B T)$ the inverse temperature. Substituting the definition of $F_B$ proves the first claimed inequality. The second one follows from the additional estimate
\begin{eqnarray*}
   \log Z_B&=&\log\sum_{i=1}^n e^{-\beta E_i^B}\leq \log\sum_{i=1}^n e^{-\beta E_{\min}^B}\\
   &=&\log n -\beta E_{\min}^B,
\end{eqnarray*}
where $E_i^B$ denotes the energy eigenvalues of $H_B$.
\end{proof}
In the special case of the truncated harmonic oscillator, with $H_B$ as in~(\ref{eqOszi}) and thus $E_{\max}^B=(m-1)\Delta E$, we get
\begin{equation}
   T'=\frac{\Delta E}{k_B}\left(
      \log\frac{e^{-\beta\Delta E}+e^{m\beta\Delta E}}{e^{\beta \Delta E}-1}
   \right)^{-1}\approx \frac T m
   \label{eqOsziCooling}
\end{equation}
for large $m$ (to first order in $1/m$).

For weakly interacting multi-particle baths, the denominators in the lower bounds for $T'$ in Theorem~\ref{TheThird} will be extensive, that is, proportional to the particle number $N$. In physically realistic scenarios, only a polynomial number $N={\rm poly}(t)$ of particles can interact with the qubit system within a finite time interval $[0,t]$. Supplemented by these plausible physical assumptions, Theorem~\ref{TheThird} tells us that
\[
   T'\gtrsim 1/t^r,
\]
where $t$ is the time in which the given thermodynamical protocol is supposed to be completed, and $r>0$ is some model-dependent exponent. This result is compatible with the findings in~\cite{MOThirdLaw}, where the authors specify the exponent $r$ in more detail, dependent on the spectrum of the heat bath. For example, they argue that $r=7$ for a radiation bath in three spatial dimensions, which is in some sense the largest physically plausible $r$. However, the approach of~\cite{MOThirdLaw} is somewhat different from the one in this paper. We will give a more detailed comparison, and some more physical intuition on our result, in the conclusions.

It is also interesting to see that the specific form of our lower bound in Theorem~\ref{TheThird} is very close to the temperature obtained via a cooling method described in~\cite{MOThirdLaw}. It relies on a different scenario that has been explored, for example, in~\cite{Browne}, which is a variation of a traditional bit reset protocol. Given any two-level system with energy gap $\Delta E$, raise the energy of the higher level isothermally while in contact with a heat bath of temperature $T$. Thermalization will bring the two-level system closer to its ground state. Traditionally, one would raise the energy level up to infinity, but suppose we stop at some $E_{\max}$, and subsequently reduce the higher energy level to its initial value. The resulting temperature of the qubit turns out to be $(T\Delta E)/(\Delta E + E_{\max})$, which is comparable to the bound in Theorem~\ref{TheThird} (and in particular~(\ref{eqOsziCooling})).

\subsection{General systems with a finite heat bath}
\label{SubsecGeneral}
In the previous subsection, we have seen by example of a qubit that the classical target states which can be obtained with a fixed finite heat bath (in particular, a truncated harmonic oscillator) are the \emph{convex hull} of those that can be achieved by classical permutation of energy levels on the same heat bath. We will now prove that this is true in general, up to some minor amendments.

To this end, let us introduce some notation. For any given initial classical state $p$  (with corresponding diagonal quantum state $\hat p$) and Hamiltonians $H_A$, $H_B$, we denote by $\mathcal{T}_Q(p,H_A,H_B)$ the set of all classical final states $p'$ that can be obtained by a suitable choice of energy-preserving unitary $U$ on $AB$, i.e.\ $[U,H_A+H_B]=0$, such that
\begin{equation}
   \hat p'=\Tr_B\left(U(\hat p\otimes\hat\gamma_B)U^\dagger\right),
   \label{eqThermal}
\end{equation}
where $\hat\gamma_B$ is the Gibbs state on $B$ for inverse temperature $\beta$. In the case of non-degenerate $H_A$, the final state $\hat p'$ will automatically be diagonal, because thermal operations cannot create coherences. If $H_A$ exhibits degeneracies, we can always make $\hat p'$ by implementing suitable unitary transformations in the degenerate subspaces.
Furthermore, by $\mathcal{T}_C(p,H_A,H_B)$, we denote the set of final states $p'$ which can be obtained by a suitable choice of energy-preserving \emph{permutation} $\pi$ on $AB$ such that
\begin{equation}
   \hat p'=\Tr_B\left(\pi(\hat p\otimes\hat\gamma_B)\pi^\dagger\right)
   \label{eqClass}
\end{equation}
(which will always result in diagonal states).
In other words, $\mathcal{T}_C(p,H_A,H_B)$ denotes the set of states that can be created from $p$ \textbf{classically} (which is a finite set), whereas $\mathcal{T}_Q(p,H_A,H_B)$ denotes those states that can be obtained \textbf{quantum-mechanically}. (These sets also depend on the inverse temperature $\beta$, which we regard as fixed in this section.) We have the following result.

\begin{theorem}
\label{TheQC}
	We have
	\begin{equation}
	   \mathcal{T}_Q(p,H_A,H_B)\subseteq {\rm conv}\, \mathcal{T}_C(p,H_A,H_B).
	   \label{eqQSubsC}
	\end{equation}
	If furthermore $H_A$ is non-degenerate, then we have
		\begin{equation}
	   \mathcal{T}_Q(p,H_A,H_B)= {\rm conv}\, \mathcal{T}_C (p,H_A,H_B).
	   \label{eqCEqualQ}
	\end{equation}
	If $H_A$ is degenerate, convex combinations of classical transitions can still be implemented exactly quantumly by adding an extra heat bath $C$ with trivial Hamiltonian $H_C=0$:
	\[
	   {\rm conv}\, \mathcal{T}_C(p,H_A,H_B) \subseteq \mathcal{T}_Q(p,H_A,H_B+0_C),
	\]
	where $\dim\, C=1+\sum_i(d_i-1)\leq \dim A$, and $d_i$ denotes the dimensions of the (energy) eigenspaces of $H_A$.
\end{theorem}
An illustration of Theorem~\ref{TheQC} (the non-degenerate case) is given in Figure~\ref{fig_tqc}. Note that the dimension of $C$ could be chosen even smaller if it turned out that a generalization of quantum Horn's lemma along the lines of Lemma~\ref{Lem23} was possible in general (for example, if Lemma~\ref{LemSqrt} turned out to be tight).
\begin{figure}[!hbt]
\begin{center}
\includegraphics[angle=0, width=.48\textwidth]{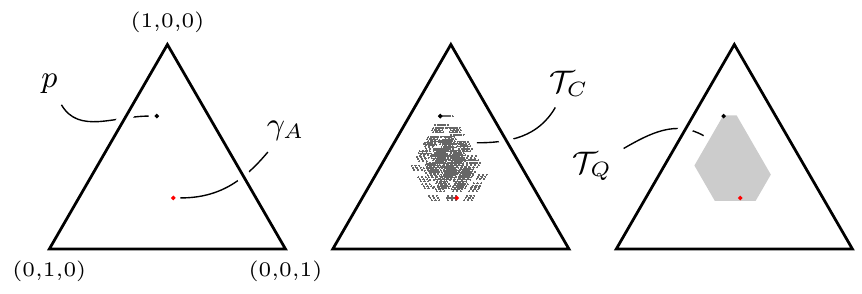}
\caption{For a given initial state $p$ in the classical probability simplex, given system Hamiltonian $H_A$, and bath Hamiltonian $H_B$, the set of states $\mathcal{T}_C(p,H_A,H_B)$ that can be obtained by classical thermal operations is a discrete set of points. In this picture we see an example calculation of this set for a three-level system where $H_A$ is non-degenerate and $H_B = H_A+H_A$ (middle triangle). However, according to Theorem~\ref{TheQC}, since $H_A$ is non-degenerate, the set of states that can be obtained via thermal quantum operations $\mathcal{T}_Q(p,H_A,H_B)$ is the convex hull of this set of states (right triangle), which is arguably a simpler and physically more realistic description. (Data used: $p=(65,22,13)/100$, $\gamma_A = (25,35,40)/100$, $\gamma_B=\gamma_A\otimes\gamma_A$.)}
\label{fig_tqc}
\end{center}
\end{figure}

\begin{proof}
Let $p'\in \mathcal{T}_Q(p,H_A,H_B)$, and let $U$ be a unitary such that~(\ref{eqThermal}) holds. Since $[U,H_A+H_B]=0$, we can decompose the total Hilbert space $AB$ into a direct sum $\bigoplus_i S_i$, where the $S_i$ are the energy eigenspaces of $H_A+H_B$. Then $U=\bigoplus_i U_i$, where the $U_i$ are block matrices acting on the $S_i$.

Let $\lambda:=p\otimes \gamma_B$. Since $\hat\lambda$ is diagonal in the energy eigenbasis, we can decompose $\hat\lambda=\bigoplus_i \hat\lambda_i$ accordingly, which means that $\lambda=\bigoplus_i \lambda_i$. Due to Lemma~\ref{LemUnistoch}, we have
\begin{equation}
U\hat\lambda U^\dagger = \bigoplus_i U_i \hat\lambda_i U_i^\dagger
= \bigoplus_i \left(\left(D_i \lambda_i\right)^\wedge  +\Omega_i\right),	
\label{eqUUStoch}
\end{equation}
where $\Omega_i$ are matrices with all diagonal elements zero, and $D_i:= U_i \circ U_i^*$ are bistochastic matrices. Let $\Omega:=\bigoplus_i \Omega_i$, then all diagonal elements of $\Omega$ are zero, i.e.\ $\langle jk|\Omega|jk\rangle=0$ for all $j,k$, and thus
\[
   \langle j|\Tr_B\Omega|j\rangle=0\mbox{ for all }j.
\]
Every $D_i$ is a bistochastic matrix acting on $S_i$, and thus, according to Birkhoff's Theorem~\cite{Bhatia}, a convex combination of permutations $\pi_i^{(j)}$ acting on $S_i$. Therefore, according to Lemma~\ref{LemConvDirectSum} in the appendix, $D:=\bigoplus_i D_i$ is a convex combinations of direct sums of permutations $\pi_i^{(k)}$ that act on the $S_i$, i.e.
\[
   D=\sum_k \mu_k\bigoplus_i \pi_i^{(k)}, \qquad \mu_k\geq 0,\enspace\sum_k\mu_k=1.
\]
Let $\pi^{(k)}:=\bigoplus_i \pi_i^{(k)}$, which is an energy-preserving permutation due to the block matrix form. We have
\begin{eqnarray*}
   p'_j &=& \langle j|\Tr_B U \hat\lambda U^\dagger |j\rangle = \langle j|\Tr_B\left[ (D\lambda)^\wedge\right]|j\rangle\\
   &=& \sum_k \mu_k \langle j|\Tr_B\left[\left(\pi^{(k)}\lambda\right)^\wedge \right]|j\rangle \\
   &=& \sum_k \mu_k \langle j|\Tr_B \left[\pi^{(k)} \hat\lambda \left(\pi^{(k)}\right)^\dagger\right] |j\rangle,
\end{eqnarray*}
and since $\hat p'$ is diagonal, this leads to
\begin{equation}
   \hat p'=\sum_k \mu_k \Tr_B\left[\pi^{(k)}(\hat p\otimes \hat\gamma)\left(\pi^{(k)}\right)^\dagger\right]
   \label{eqConvClass}
\end{equation}
which is a convex combination of outputs of maps of the form~(\ref{eqClass}). Thus $p'\in{\rm conv}\,\mathcal{T}_C(p,H_A,H_B)$, which proves~(\ref{eqQSubsC}). 

Now we would like to show the second half of the theorem, that every $p'\in{\rm conv}\,\mathcal{T}_C(p,H_A,H_B)$ can be obtained quantum-mechanically.
So suppose there are energy-preserving permutations $\pi^{(k)}$ and coefficients $\mu_k$ such that~(\ref{eqConvClass}) holds. Energy preservation implies that $\pi^{(k)}=\bigoplus_i \pi_i^{(k)}$, where the $\pi_i^{(k)}$ act on the global energy eigenspaces $S_i$. Let us define 
\[
\hat\lambda' = \sum_k \mu_k \ \pi^{(k)}\hat p\otimes \hat\gamma\left(\pi^{(k)}\right)^\dagger = \sum_k \mu_k\  \pi^{(k)}\hat \lambda\left(\pi^{(k)}\right)^\dagger
\]
which gives $ \hat p'=\Tr_B \hat\lambda'$. As before, we can decompose $\lambda=\bigoplus_i \lambda_i$ and $\lambda'=\bigoplus_i \lambda'_i$ into blocks corresponding to the global energy eigenspaces, and it follows that $\lambda'_i = \sum_k \mu_k \pi_i^{(k)} \lambda_i$, and so $\lambda_i\succ \lambda'_i$, i.e.\ $\lambda_i$ majorizes $\lambda'_i$ by a general result on majorization~\cite{MarshallOlkin}. Thus, according to~\cite[Thm.\ B.6]{MarshallOlkin}, there are unistochastic maps $D_i$ such that $\lambda'_i=D_i \lambda_i$ (note that $D_i$ need not be equal to $\sum_k \mu_k \pi_i^{(k)}$), and therefore unitary matrices $U_i$ such that $D_i=U_i\circ U_i^*$. 
We set $U:=\bigoplus_i U_i$ and by Lemma~\ref{LemUnistoch}, as in~(\ref{eqUUStoch}), we will have 
\[
   U\hat\lambda U^\dagger=\hat\lambda'+\Omega,
\]
where $\Omega=\bigoplus \Omega_i$ is a matrix with all diagonal elements zero. It follows that
\[
   \langle j|\Tr_B\left(U\hat\lambda U^\dagger\right)|j\rangle=\langle j|\Tr_B\hat\lambda'|j\rangle+\sum_k\langle jk|\Omega|jk\rangle=p'_j.
\]
For thermal operations, we know that $\Tr_B\left(U\hat\lambda U^\dagger\right)$ must be block-diagonal~\cite{HO2013} in the energy eigenbasis due to energy conservation. Thus, if $H_A$ is non-degenerate, it follows that this matrix is diagonal in the unique energy eigenbasis, and we have
\[
   \hat p'=\Tr_B\left[U(\hat p\otimes \hat\gamma)U^\dagger\right],
\]
which proves~(\ref{eqCEqualQ}). However, if $H_A$ is degenerate, we do not necessarily obtain a diagonal matrix. We deal with this by introducing an auxiliary system used to implement a decoherence map that sets the off-diagonal elements of $\rho'_A:=\Tr_B(U\hat\lambda U^\dagger)$ to zero. Since $\rho'_A$ is block-diagonal in  the energy eigenbasis, with block sizes (at most) $d_i$, Lemma~\ref{LemThermalDecoherence} below shows that we can implement the decoherence map by adding a fully degenerate auxiliary system of dimension $1+\sum_i(d_i-1)$ and performing a suitable thermal operation.
\end{proof}
Note that this theorem contains Lemma~\ref{LemNoisyN} from Section~\ref{SecNoisy} as a special case: in the case of a trivial $n\times n$ Hamiltonian $H_A=0$, we have a single degeneracy $d_1=n$. Due to Birkhoff's Theorem~\cite{Bhatia}, all bistochastic maps can be implemented as convex combinations of classical permutations on system $A$ alone. Therefore, denoting a trivial heat bath $B$ of dimension $\dim B=1$ (and arbitrary Hamiltonian $H_B\in\R$) simply by a bullet, we have
\[
   {\rm conv}\, \mathcal{T}_C(p,H_A,\bullet)=\{p'\,\,|\,\, p\succ p'\}.
\]
Then Theorem~\ref{TheQC} tells us that we can implement all these transitions exactly by noisy quantum operations with a heat bath $C$ of size
\[
   \dim\, C = 1+(d_1-1)=n.
\]

The following lemma has been used in the proof of Theorem~\ref{TheQC}.
\begin{lemma}[Thermal decoherence]
\label{LemThermalDecoherence}
	Let $A$ be an $n$-dimensional quantum system with Hamiltonian $H_A$, and $\{|i\rangle\}_{i=1}^n$ be an arbitrary orthonormal eigenbasis of $H_A$. Furthermore, suppose that $\rho_A$ is a block-diagonal density matrix on $A$, i.e.\ $[\rho_A,H_A]=0$. Then there is a heat bath $B$ with trivial Hamiltonian $H_B=0$ of dimension
	\[
	   \dim B =1+\sum_i(d_i-1)\leq n,
	\]
   where $d_i$ denotes the dimensions of the energy eigenspaces of $H_A$, such that there is a thermal operation on $A$ with this heat bath  $B$ that maps $\rho_A$ to $\rho'_A$, where $\langle i|\rho'_A|i\rangle=\langle i|\rho_A|i\rangle$, and $\langle i|\rho'_A|j\rangle=0$ for $i\neq j$.
\end{lemma}
In other words: small heat baths allow to decohere block-diagonal quantum states. If $H_A$ is non-degenerate, i.e.\ all $d_i=1$, then $\rho_A$ is already diagonal and we do not need any heat bath at all, i.e.\ $\dim B=1$.
\begin{proof}
	We construct a thermal operation $\Phi$ on $A$ which sets suitable matrix elements of the input density matrix $\rho$ to zero. Suppose we are given any subset $S\subset \{1,2,\ldots, n\}$, then we will construct $\Phi$ such that it sets all elements in rows and columns appearing in $S$ to zero (except for the diagonal elements) and leaves all other matrix elements invariant. For example, suppose that $n=4$, and that $S=\{2,4\}$. Then we will construct a thermal operation $\Phi$ such that
	\[
	   \left(\begin{array}{ccc} \rho_{11} & \hdots & \rho_{14} \\	
	   \vdots & & \vdots \\
	   \rho_{41} & \hdots & \rho_{44}
\end{array}\right)\stackrel \Phi \mapsto
\left(\begin{array}{cccc}
   \rho_{11} & 0 & \rho_{13} & 0 \\
   0 & \rho_{22} & 0 & 0 \\
   \rho_{31} & 0 & \rho_{33} & 0 \\
   0 & 0 & 0 & \rho_{44}\end{array}
\right).
	\]
	Clearly, if $\rho$ is already block-diagonal with block sizes $d_i$ (e.g.\ in this case $d_1=2$ and $d_2=2$), then this map fully decoheres $\rho$, as long as we decohere $(d_i-1)$ rows and columns in every $i$-th block. In other words, we will choose the subset $S$ such that
	\[
	   |S|=\sum_i(d_i-1),
	\]
	and if $\Phi$ acts as stated, we achieve the decoherence as claimed in the statement of the lemma. Let $B=\C^d$, $d=|S|+1$, and $|1\rangle,\ldots,|d\rangle$ some orthonormal basis on $B$. Define the permutation matrix $\pi$ by $\pi|i\rangle=|i+1\rangle$ if $1\leq i \leq d-1$ and $\pi|d\rangle=|1\rangle$. We have $\pi^T=\pi^\dagger=\pi^{-1}$, and, for $j\in\mathbb{Z}$,
	\[
	   \tr(\pi^j) = \left\{
	      \begin{array}{cl}
	      d & \mbox{if }	 j\in\{\ldots,-2d,-d,0,d,2d,\ldots\}\\
	      0 & \mbox{otherwise}.
	      \end{array}
	   \right.
	\]
	Label the elements of $S$ (in increasing order) by $S=\{s_1,s_2,\ldots,s_{d-1}\}$, and define the unitary $U$ on $AB$ by
	\[
	   U=\sum_{i\not\in S} |i\rangle\langle i|\otimes\mathbf{1}+\sum_{j=1}^{|S|} |s_j\rangle\langle s_j|\otimes \pi^j.
	\]
	Clearly $[U,H_A\otimes\mathbf{1}]=0$, so if we regard $B$ as a quantum system with trivial Hamiltonian $H_B=0$, the following map $\Phi:\rho\mapsto\rho'$  is a thermal operation:
	\[
	   \rho'=\Tr_B\left[ U\left(\rho\otimes \frac{\mathbf{1}} d \right)U^\dagger\right].
	\]
	It is easy to check by direct calculation that
	\[
	   \rho'_{ij}=\left\{
	      \begin{array}{cl}
	      \rho_{ij} & \mbox{if } i=j \mbox{ or }(i\not\in S \mbox{ and }j\not\in S),\\
	      0 & \mbox{otherwise},
	      \end{array}
	   \right.
	\]
	where $\rho_{ij}=\langle i|\rho|j\rangle$, and similary for $\rho'_{ij}$. For example, if $i\not\in S$ and $j\in S$, then $\langle ik|U=\langle ik|$, and
	\[
	   U^\dagger |jk\rangle = \sum_{m=1}^{|S|} |s_m\rangle\langle s_m|j\rangle \otimes \pi^{-m} |k\rangle=|j\rangle\otimes \pi^{-m} |k\rangle,
	\]
	where $m$ is such that $s_m=j$. Thus, we get
	\begin{eqnarray*}
		\rho'_{ij}&=& \frac 1 d \sum_k \langle ik|\rho\otimes\mathbf{1}|j\rangle\otimes \pi^{-m}|k\rangle\\
		&=&\frac 1 d \rho_{ij} \sum_k \langle k|\pi^{-m}|k\rangle = \rho_{ij}\frac 1 d \tr(\pi^{-m})=0.
	\end{eqnarray*}
	All other cases of $i$ and $j$ can be checked similarly.
\end{proof}

\begin{figure}[!hbt]
\begin{center}
\includegraphics[angle=0, width=.48\textwidth]{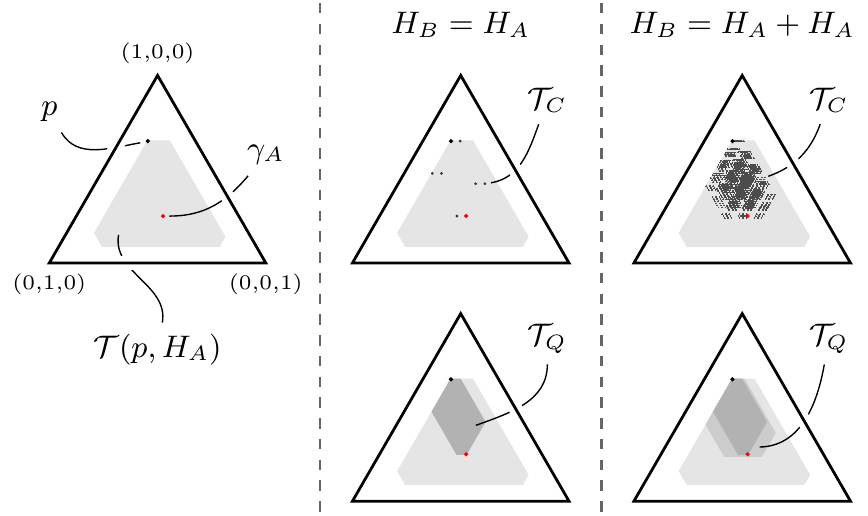}
\caption{The set of (classical) states in the probability simplex which are thermomajorized by $p$, denoted $\mathcal{T}(p,H_A)$, is a compact convex set (light gray). If $H_A$ is non-degenerate, then the set of states $\mathcal{T}_Q(p,H_A,H_B)$ that one can obtain with a fixed, finite heat bath, by quantum thermal operations, is a compact convex subset (dark gray, second row). With increasing resources (for example, by having several copies of the heat bath), this subset will grow (third column), and finally cover every point in the relative interior of $\mathcal{T}(p,H_A)$. However, some points on the boundary will never be covered (we have proven this for the qubit in Lemma~\ref{LemHarmonicOsci}, and conjecture that it is also true for higher-level systems). This can be interpreted as an instance of the third law of thermodynamics.}
\label{fig_interior}
\end{center}
\end{figure}

Theorem~\ref{TheQC} implies in particular the following result, which is illustrated in Figure~\ref{fig_interior}.
\begin{corollary}
\label{CorInterior}
Let $p$ be any classical state on a quantum system with some Hamiltonian $H_A$, and $\mathcal{T}(p,H_A)$ the (compact convex) set of all classical states that are thermomajorized by $p$. Then all states $p'\in{\rm relint}\, \mathcal{T}(p,H_A)$ (i.e.\ in the relative interior) can be obtained exactly by quantum thermal operations with finite heat baths, but not in general all states in the relative boundary  ${\rm rebd}\,\mathcal{T}(p,H_A)$.
\end{corollary}

\begin{proof}
We use Lemma~\ref{LemJanzing} which shows that all states $p'\in\mathcal{T}(p,H_A)$ can be obtained to arbitrary accuracy by classical thermal operations with suitable heat baths. Formally, if we denote the set of all states that can be obtained from $p$ classically with any possible heat bath by
\[
   \mathcal{T}_C(p,H_A):=\bigcup_{\dim B\in\mathbb{N}, H_B} \mathcal{T}_C(p,H_A,H_B),
\]
then $\mathcal{T}_C(p,H_A)$ is dense in $\mathcal{T}(p,H_A)$. It follows with a bit of convex geometry (Lemma~\ref{LemConvexGeometry} in the appendix) that
\[
   {\rm relint}\,\mathcal{T}(p,H_A)\subseteq {\rm conv}\,\mathcal{T}_C(p,H_A).
\]
By construction, if $p'\in {\rm conv}\,\mathcal{T}_C(p,H_A)$, then there exist $m\in\mathbb{N}$, $\lambda_1,\ldots,\lambda_m\geq 0$ with $\sum_i \lambda_i=1$, heat baths $B_i$ with Hamiltonians $H_{B_i}$, and states $q^{(i)}$ on $A$ such that $p'=\sum_{i=1}^m \lambda_i q^{(i)}$ and $q^{(i)}\in\mathcal{T}_C(p,H_A,H_{B_i})$. If we have a composite heat bath on the tensor product $B_1\otimes\ldots \otimes B_m$ with Hamiltonian $H_{B_1}+\ldots+H_{B_m}$, then
\[
   \mathcal{T}_C(p,H_A,H_{B_i})\subseteq \mathcal{T}_C(p,H_A,H_{B_1}+\ldots+H_{B_m}),
\]
since we can always perform a classical permutation that ignores (acts as the identity on) some of the tensor factors.
Thus, $q^{(i)}\in \mathcal{T}_C(p,H_A,H_{B_1}+\ldots+H_{B_m})$, and so $p'\in {\rm conv}\,\mathcal{T}_C(p,H_A,H_{B_1}+\ldots+H_{B_m})$ which, according to Theorem~\ref{TheQC}, is contained in $\mathcal{T}_Q(p,H_A,H_{B_1}+\ldots+H_{B_m}+0_C)$, with $C$ a trivial heat bath of suitable finite size. Therefore, all elements of ${\rm relint}\,\mathcal{T}(p,H_A)$ can be obtained exactly by quantum thermal operations with finite heat baths.

On the other hand, states in the relative boundary cannot always be obtained exactly. This is shown by example in Lemma~\ref{LemHarmonicOsci}.
\end{proof}

\section{Conclusions}
\label{SecConclusions}
The study of quantum operations with finite heat baths reveals some interesting mathematical structure and physical insights. In this paper, we have analyzed basic properties of noisy and thermal operations with finite heat baths. In order to study whether some state transitions are unattainable with finite baths (resembling the third law of thermodynamics), we first had to remove an unphysical apparent``unattainability'' that is simply a mathematical artefact of a semiclassical treatment: the fact that infinite auxiliary systems are needed for some simple state transitions if one restricts to classical permutations in the energy eigenbasis.

We did this by proving Corollary~\ref{CorInterior}, which is also sketched in Figure~\ref{fig_interior}: the target states that can be obtained from a given (classical) initial state with unbounded heat baths is a compact convex set. All the states in its interior can be obtained exactly via thermal quantum operations with finite heat baths. It is then (some of) the states in the boundary of this set that become the subject of potential unattainability results.

In the case of a two-level quantum system, this boundary consists of two states (cf.\ Figure~\ref{fig_oscillator}): the initial state $p$ itself, and another state $p^*$ that is in some sense ``thermodynamically antipodal'' to $p$. In fact, we have shown that $p^*$ can never be obtained exactly from $p$ with any finite heat bath. In the case of a pure excited initial state, this led us to an instance of the third law of thermodynamics (Theorem~\ref{TheThird}), giving a (physically meaningful) lower bound on the temperature to which the system can be cooled by using any finite-dimensional heat bath. While we have formally proven this result only for qubits, it seems plausible that an analogous statement holds for higher-dimensional quantum systems, too. This gives an intuitive structural picture of unattainability in terms of the geometry of the state space (cf.\ Figure~\ref{fig_interior}), and it suggests a generalization of the third law to situations beyond cooling to absolute zero (since $p^*$ is not in general the ground state).

We can elucidate the relation to other recent unattainability results by elaborating on the physical intuition of our result. Consider cooling with thermal operations, as in the scenario of Figure~\ref{fig_cooling}. Denoting the ground state of the qubit by $|1\rangle$ and the excited state by $|2\rangle$, we would like to apply an energy-preserving unitary to the global initial state $|2\rangle\langle 2|\otimes\hat\gamma_B$ in a way that \emph{moves as much probability weight as possible} (in fact, all of it) into the local ground state, i.e.\ the subspace spanned by the global eigenstates $|1i\rangle$. However, if $|m\rangle$ denotes the maximal energy level of the bath (assuming for the moment that it is unique), then the (globally maximal) energy level $|2m\rangle$ is ``trapped'': energy preservation forbids to move its probability weight into any other subspace. And since the bath starts in a thermal state, this ``trapped'' population will be non-zero, and it will not be transferred into the ground state. We expect that this simple intuition will remain robust also in more complicated cases beyond the simple qubit scenario, for example in cases where the quantum system $A$ is itself composite, $A=A_1 A_2$. This encompasses in particular situations in which $A_1$ is the system of interest (say, a quantum system to be cooled), while $A_2$ is a work storage system, such as a weight~\cite{Skrzypczyk} or work bit~\cite{HO2013}.

This ``inability to move probability weight around'' in high-lying energy levels of the heat bath is also a major ingredient of the version of the third law proven in~\cite{MOThirdLaw}. In our setting, however, this inability comes from an assumption of finite-dimensionality of the bath, while in~\cite{MOThirdLaw}, it is the bath's spectral density together with the assumption of a strict upper bound on the available work $w_{\max}$ that limits the available energy levels. Our approach should not be understood as claiming that physical heat baths are generically finite-dimensional. Instead, the motivation is to give a very general formulation of the idea that we can only engineer interactions with a finite number of energy levels of the bath in a controlled way within a thermodynamic protocol. If understood in this way, the approach of~\cite{MOThirdLaw} is conceptually related to ours (but~\cite{MOThirdLaw} goes into much more physical detail).

Our unattainability statement (and the one in~\cite{MOThirdLaw}) is also different from the scenario studied by Janzing et al.\ in~\cite{Janzing}, where they consider the amount of non-thermal \emph{resource states} needed to cool a system (close) to absolute zero. In contrast, our formulation assumes that all non-thermal resources are in principle there (e.g.\ the energy of the excited state in Figure~\ref{fig_cooling} is enough for cooling to zero), but they can only be harvested by thermodynamic protocols in the presence of an unbounded heat bath.

Several further recent publications have considered limitations for cooling that arise from finite heat baths, cf.~\cite{Reeb} and~\cite{Mahler2011} (see also~\cite{Wu} and~\cite{Ticozzi}), but within approaches that are not formulated within the resource-theoretic framework of thermal operations. Instead, arbitrary (in general energy-non-preserving) unitaries are allowed, and the energy cost or yield is taken into account in retrospect, by bookkeeping the system's average energy content. That is, the fundamental microscopic reversibility of the physical laws is taken into account directly by restricting to unitary transformations, but the principle of total energy conservation is only considered indirectly.

In the context of the third law, this leads to a subtle restriction that comes from the resulting focus on entropic properties. Typically, in approaches that allow arbitrary unitaries, the impossibility of cooling to absolute zero is ultimately attributed to a simple property of linear algebra: namely, that \emph{unitary operations cannot create zero eigenvalues from non-zero eigenvalues}. In contrast, our setup yields more severe restrictions that rely on microscopic reversibility \emph{in conjunction with} total energy preservation. The simplest example is the one depicted in Figure~\ref{fig_cooling}: even though we start with a pure state of the system, the pure ground state is still unattainable.

Our paper leaves several interesting open questions. First, we have shown by example in Lemma~\ref{Lem23} that quantum Horn's lemma is true (Lemma~\ref{LemNoisyN}), but not optimal: for example, all noisy state transitions on a three-dimensional quantum system can be accomplished by means of a two-dimensional auxiliary system. What is the smallest possible bath dimension in general? Is the bound in Lemma~\ref{LemSqrt} tight? If so, then this would also improve our more general result for non-trivial Hamiltonians, Theorem~\ref{TheQC}.	 As we have shown, this question is closely related to a version of the quantum marginal problem.

From the point of view of the physical interpretation, the most interesting open problem is to analyze in more detail whether, and if so in what way, our third law for qubits (Theorem~\ref{TheThird}) generalizes to higher-dimensional quantum systems and scenarios with a work storage device in a physically meaningful way. Restricting the dimension of the heat bath, as we have done here, is only one of several possibilities to study unattainability results in thermodynamics. But the picture of some states, not necessarily of temperature zero, that lie ``on the edge'' of the set of attainable states (as in Figure~\ref{fig_interior}) is a quite suggestive hint towards possible generalizations beyond perfect cooling.

\section*{Acknowledgments}
We are grateful to Manfred Salmhofer for discussions, and we would like to thank Michael Walter and Matthias Christandl for drawing our attention to~\cite{LiPoonWang} and~\cite{Vergne}.

\onecolumngrid

\section{Appendix}
The following lemma gives a partial classification of the set of \emph{maps} that can be implemented via noisy operations on any $n$-dimensional system $A$ with an $n$-dimensional bath $B$. Since $p\succ p'$ implies that there is a unistochastic map $D$ with $Dp=p'$~\cite[Thm.\ B.6]{MarshallOlkin}, the lemma below implies Lemma~\ref{LemNoisyN} as a simple corollary. It can therefore be interpreted as yet another generalization of quantum Horn's lemma.
\begin{lemma}
\label{Lem3}
Let ${\rm Uni}(n)$ and ${\rm Bi}(n)$ be the sets of uni\-stochastic resp.\ bistochastic maps on $\R^n$, i.e.
\begin{eqnarray*}
	{\rm Uni}(n)&:=& \left\{ D\in\R^{n\times n}\,\,|\,\, \exists U: U^\dagger U=\mathbf{1}, D_{ij}=|U_{ij}|^2\right\},\\
	{\rm Bi}(n)&:=&\left\{ D\in\R^{n\times n}\,\,\left|\,\, \sum_i D_{ij}=\sum_j D_{ij}=1, D_{ij}\geq 0\right.\right\}.
\end{eqnarray*}
Moreover, let ${\rm Noisy}(n,n)$ be the set of maps $p\mapsto p'$ on $\R^n$ which can be written in the form
\[
   \hat p'=\Tr_B\left[ U \left( \hat p \otimes \frac{\mathbf{1}}n\right) U^\dagger\right]
\]
with some unitary $U$ on $\C^n\otimes \C^n$. Then we have
\begin{eqnarray*}
   {\rm Uni}(2) &=& {\rm Noisy}(2,2)={\rm Bi}(2),\\
   {\rm Uni}(n)&\subsetneq& {\rm Noisy}(n,n)\subseteq {\rm Bi}(n) \qquad \mbox{for }n\geq 3.
\end{eqnarray*}
Moreover, the inclusion ${\rm Noisy}(n,n)\subseteq {\rm Bi}(n)$ is strict if and only if ${\rm Noisy}(n,n)$ is not convex.
\end{lemma}
\textbf{Remark.} We do not currently know whether ${\rm Noisy}(n,n)$ is convex for $n\geq 3$, which would imply that ${\rm Noisy}(n,n)={\rm Bi}(n)$.
\begin{proof}
If $D\in{\rm Uni}(n)$ then there is a unitary matrix $V$ such that $D_{ij}=|V_{ij}|^2$. Thus, for any state $p\in\R^n$, we have
\[
   \langle i|V\hat p V^\dagger |i\rangle=\sum_j V_{ij} p_j V_{ij}^*=(Dp)_i;
\]
in other words, $V\hat p V^\dagger$ has the vector $Dp$ on its diagonal. We can compose $V$ with the unitary $W$ from the proof of Lemma~\ref{LemNoisyN} which implements the decoherence map to remove all off-diagonal elements with an $n$-dimensional bath. Then $U=W(V\otimes\mathbf{1}_B)$ will realize the map $D$ as a noisy operation. This proves that ${\rm Uni}(n)\subseteq {\rm Noisy}(n,n)$. On the other hand, for any given operation in ${\rm Noisy}(n,n)$ with associated unitary $U$, define
\[
   D_{ik}:=\frac 1 n \langle i|\Tr_B\left[ U\left(|k\rangle\langle k|\otimes\mathbf{1}\right)U^\dagger\right]|i\rangle,
\]
then it is straightforward to check that $\sum_i D_{ik}=\sum_k D_{ik}=1$ and $D_{ik}\geq 0$. Moreover, if $\hat p'$ is the output of the noisy operation on input $\hat p$, then $p'=D p$. This shows that ${\rm Noisy}(n,n)\subseteq {\rm Bi}(n)$. For $n=2$, we have ${\rm Uni}(2)={\rm Bi}(2)$, which proves the first statement of the lemma.

For $n\geq 3$, we construct a matrix $D\in{\rm Bi}(n)\setminus {\rm Uni}(n)$ which can be implemented by a noisy operation. Let $\pi$ be the cyclic permutation on $\C^n$ from the proof of Lemma~\ref{LemNoisyN}, and define
\[
   D:=\left(1-\frac 1 n\right)\mathbf{1}+\frac 1 n \pi.
\]
Obviously $D$ is bistochastic. To see that it is not uni\-stochastic, suppose $U$ is a unitary matrix with $D_{ij}=|U_{ij}|^2$, then orthogonality of the first two rows of $U$ is $\sum_{j=1}^n U_{1j}^* U_{2j}=0$. But for every $j\neq 1$, either $D_{1j}=0$ or $D_{2j}=0$, hence this sum equals $U_{11}^* U_{21}=e^{i\theta}\sqrt{D_{11}D_{21}}=e^{i\theta}\sqrt{\left(1-\frac 1 n\right)\frac 1 n}$ for some $\theta\in\R$, and this is not zero. Define the unitary
\[
   U=\mathbf{1}\otimes \sum_{k=2}^n |k\rangle\langle k|+\pi\otimes |1\rangle\langle 1|.
\]
A simple calculation shows that this unitary generates a noisy operation which maps $\hat p$ to $(Dp)^\wedge$. This shows that $D\in{\rm Noisy}(n,n)$, but $D\not\in{\rm Uni}(n)$.

Finally, since ${\rm Bi}(n)$ is convex, non-convexity of ${\rm Noisy}(n,n)$ would imply that ${\rm Noisy}(n,n)\subsetneq {\rm Bi}(n)$. Conversely, suppose that this strict inclusion was true (note that we do not know whether this is the case). Since every permutation is unistochastic, and the bistochastic maps are the convex hull of the permutations due to Birkhoff's Theorem~\cite{Bhatia}, ${\rm Noisy}(n,n)$ could then not be convex, because otherwise the permutations that it contains would generate all bistochastic maps.
\end{proof}
We conjecture that ${\rm Noisy}(n,n)$ is not convex (and thus a strict subset of the bistochastic maps), but we were not able to show this. This contrasts the simple behavior of the set of transitions that can be realized with an $n$-dimensional bath, characterized by majorization.

The following lemma has been used in the proof of Theorem~\ref{TheQC}.
\begin{lemma}
\label{LemConvDirectSum}
Let $B_1,\ldots,B_n$ be matrices such that every $B_i$ is a convex combination of other matrices $B_i^{(j)}$. Then
\[
   B=B_1\oplus\ldots\oplus B_n
\]
is a convex combination of the matrices
\[
   B_1^{(j_1)}\oplus\ldots\oplus B_n^{(j_n)}.
\]
\end{lemma}
\begin{proof}
The case $n=1$ is trivial, and the cases $n\geq 3$ follow from repeated application of the statement for $n=2$, so it is sufficient to prove the lemma for $n=2$. Suppose that
\[
   B_1=\sum_i \lambda_i B_1^{(i)},\quad B_2=\sum_j \mu_j B_2^{(j)}, \quad \lambda_i,\mu_j\geq 0, \quad \sum_i\lambda_i=\sum_j\mu_j=1.
\]
Then
\[
   B_1\oplus B_2 = \left(\sum_{ij} \lambda_i \mu_j B_1^{(i)}\right)\oplus \left(\sum_{ij}\lambda_i \mu_j B_2^{(j)}\right)=\sum_{ij} \lambda_i \mu_j \left(B_1^{(i)}\oplus B_2^{(j)}\right).
\]
Note that this statement can easily be illustrated by writing down the block matrices concretely.
\end{proof}

The following lemma concerns the relative interior of a convex set~\cite{Webster}. It has been used in the proof of Corollary~\ref{CorInterior}.
\begin{lemma}
\label{LemConvexGeometry}
Let $C\subset\R^n$ be convex, and $X\subset C$ a dense subset of $C$. Then ${\rm relint}\, C \subseteq {\rm conv}\,X$.
\end{lemma}
\begin{proof}
Denote the $\varepsilon$-ball around some point $c\in\R^n$ by $B_{\varepsilon}(c)=\{x\in\R^n\,\,|\,\,\|x-c\|\leq\varepsilon\}$.
Let $c\in{\rm relint}\, C$, then there is $\delta>0$ small enough such that $B_\delta(c)\cap {\rm aff}(C)\subset C$. Since $B_{\delta}(c)\cap {\rm aff}(C)$ is an $m$-dimensional Euclidean unit ball, with $m=\dim C$, there is a regular $m$-simplex with vertices $v_1,\ldots,v_{m+1}$ on the surface of that ball, i.e.\ $\|v_i-c\|=\delta$, and center $c=\frac 1 {m+1}\sum_{i=1}^{m+1} v_i$. Since $X$ is dense in $C$, we can find $v'_1,\ldots,v'_{m+1}\in X$ that are arbitrarily close to the vertices $v_1,\ldots,v_{m+1}$. If they are close enough, it is geometrically clear that their convex hull must contain $c$.
\end{proof}
\begin{figure}[!hbt]
\begin{center}
\includegraphics[angle=0, width=.5\textwidth]{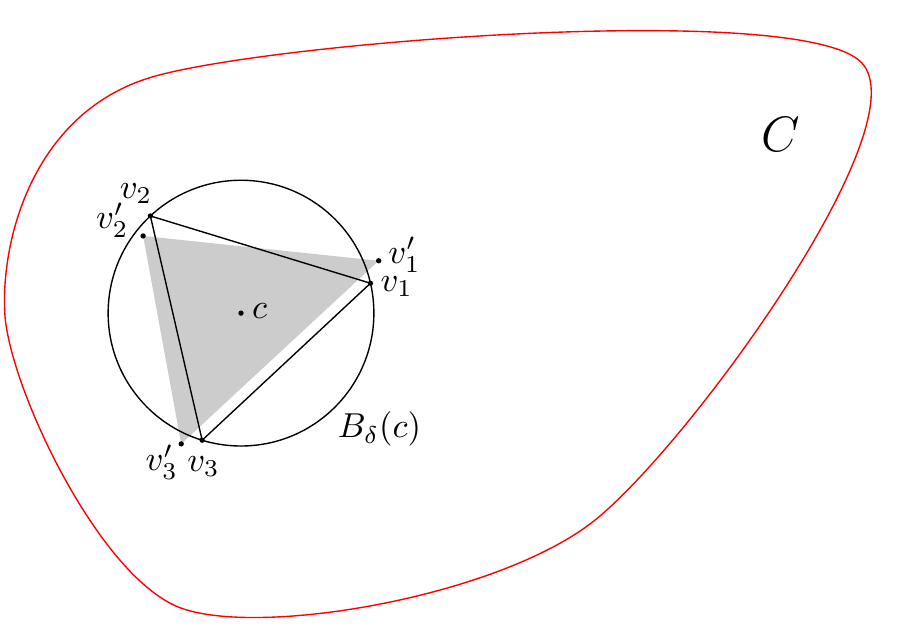}
\caption{Illustration of the proof of Lemma~\ref{LemConvexGeometry}.}
\label{fig_lemma16}
\end{center}
\end{figure}

As mentioned in the main text, Lemma~\ref{LemMarginal} has been proven independently in~\cite{LiPoonWang}, and it also follows from the results in~\cite{Vergne}. Nevertheless, we give a self-contained proof for completeness.

\newtheorem*{LemMarginal}{Lemma \ref{LemMarginal}}
\begin{LemMarginal}
Let $\rho$ be a quantum state on $AB$, and $\sigma$ a quantum state on $A$. If $\dim A\leq \dim B$, there exists a unitary $U$ on $AB$ with
\[
 \Tr_B\left(U\rho U^\dagger\right)=\sigma
\]
if and only if
\begin{equation}
   \Tr_B \hat\lambda(\rho)\succ\sigma,
\end{equation}
where $\lambda(\rho)$ is the vector of eigenvalues of $\rho$, ordered in non-increasing order, and $\hat\lambda(\rho)$ is the corresponding diagonal density matrix.
\end{LemMarginal}

\begin{proof}
Suppose there is some unitary $U$ such that $\Tr_B(U\rho U^\dagger)=\sigma$.
Let us assume without loss of generality that $	\rho$ and $\sigma$ are diagonal and their diagonal elements are non-increasingly ordered, so $\hat\lambda(\rho) = \rho$ and $\hat\lambda(\sigma) = \sigma$.
With Lemma~\ref{LemUnistoch} we can write 
$
U\hat\lambda(\rho)\ U^\dagger = \left(D \lambda(\rho) \right)^\wedge+ \Omega
$
where $D$ is a stochastic matrix with $D_{ij}= |U_{ij}|^2$ and $\Omega$ some matrix with only zeros on the diagonal. Thus, it holds $\lambda (\rho) \succ D\lambda (\rho)$, and $\Tr_B\Omega$ has also only zeros on its diagonal.  Since $\sigma$ is diagonal, the off-diagonal elements of $\Tr_B\Omega$ must vanish identically, and we get $\Tr_B \Omega=0$ and $\lambda(\sigma)=\Tr_B\left( D\lambda(\rho) \right)^\wedge$.  If $p,q$ are probability vectors on a joint system $AB$ and $p$ is non-increasingly ordered, then from $p\succ q$, it follows that $\Sigma_B p \succ \Sigma_B q$ where $\Sigma_B$ is a the marginal map with respect to the $B$ system which satisfies $\Tr_B \hat\lambda = (\Sigma_B\lambda)^\wedge.$
So we have
\[
\Tr_B \hat\lambda(\rho) \succ \Tr_B \left( D\lambda(\rho)\right)^\wedge
 = \hat\lambda(\sigma)\]
which proves the ``only if'' direction of the statement (and is independent of the dimensions of $A$ and $B$).

To prove the converse direction, set $n= \dim A, m=\dim B$. Lemma~\ref{LemUnistoch} that we have just used is basically the one half of the Schur-Horn Theorem that states that the vector of diagonal elements of a Hermitian matrix is always majorized by the vector its eigenvalues.  For the second half of the proof one main ingredient will be the other half of the theorem which states: if $\lambda$ is the vector of eigenvalues of a Hermitian matrix then for each vector $p$ with $\lambda \succ p$ there exists a basis in which $p$ is the vector of diagonal elements of this Hermitian matrix.

So assume that $\Tr_B \hat\lambda(\rho)\succ\sigma$ which means that $\Sigma_B \lambda(\rho) \succ \lambda(\sigma)$. Without loss of generality, we may assume that $\sigma$ is diagonal (a local diagonalizing unitary can always be implemented via the global unitary $U$ that we are constructing now), $\sigma = \hat\lambda(\sigma)$. By the Schur-Horn theorem we know there exist a unitary $\bar U$ on the $A$ system such that
 \[
 \bar U \left( \Tr_B \hat\lambda(\rho) \right) \bar U^\dagger
 \]
 has diagonal elements $\lambda(\sigma)$. We will use this to define a new unitary on the whole $AB$ system such that $\Tr_B\left(U\rho U^\dagger\right)=\sigma$ is satisfied. So define
 \[
 U = \sum_{i,j=1}^n  u_{ij} \ketbra{i}{j}\otimes \pi^{j-i} ,
 \]
 where $u_{ij}$ are the components of $\bar U$, and $\pi^{j-i}$ is a (later to be specified) permutation matrix $\pi$ on the $B$ system to the power $(j-i)$.   Since $\pi^\dagger = \pi^{-1}$ and $\bar U$ itself is unitary, this is a unitary matrix as well.
Now the rest of the proof is merely showing by direct calculation that $\Tr_B\left( U\hat\lambda(\rho)U^\dagger \right) = \hat\lambda(\sigma)$.

For this we write the diagonal matrix $\hat\lambda (\rho)$ as
\[
\hat\lambda (\rho) = \sum_{i=1}^n \sum_{\mu=1}^m \lambda_{i\mu} \ketbra{i}{i}\otimes\ketbra{\mu}{\mu} = \sum_{i=1}^n   \ketbra{i}{i} \otimes \Lambda_i\quad\text{with}\quad \Lambda_i = \sum_{\mu=1}^m \lambda_{i\mu} \ketbra{\mu}{\mu},
\]
and with Lemma~\ref{LemUnistoch} we write again
\[
U\hat\lambda(\rho)U^\dagger = \left( D \lambda(\rho) \right)^\wedge + \Omega
\] 
where $D$ is the unistochastic matrix $D = U\circ U^*$ and $\Omega$ some matrix with only zeroes on the diagonal.  If we define $\bar D = \bar U \circ \bar U^*$ it is simple to show by direct calculation that 
\[
\Sigma_B \bigl( D \lambda(\rho) \bigr)= \bar D \bigl( \Sigma_B \lambda(\rho) \bigr) = \lambda(\sigma) ,
\]
where the second equality sign holds by assumption. So we know that $\Tr_B \left( D \lambda(\rho) \right)^\wedge = \hat \lambda(\sigma)$ and it remains to show that $\Tr_B \Omega = 0$.

By direct calculation, starting with the form of $\hat\lambda(\rho)$ and $U$ as just given, we get the form 
\[
\Omega = \sum_{\substack{ij=1\\i\neq j}}^n \ketbra ij \otimes \sum_{k=1}^n u_{ik}\,u^*_{jk} \ \pi^{i-k}\,\Lambda_k\,\pi^{k-j} + \sum_{i=1}^n \ketbra ii \otimes \Omega^{(ii)},
\]
where $\Omega^{(ii)}= \bra i\Omega\ket{i}$ is a matrix on the $B$ system with all diagonal elements equal to zero. Since $\Tr_B \Omega^{(ii)}=0$ the second term vanishes when calculating $\Tr_B \Omega$. 
However we still have to look at the first term where we calculate 
\[
\tr \bigl( \pi^{i-k}\,\Lambda_k\,\pi^{k-j} \bigr)= \tr \bigl( \pi^{i-j}\,\Lambda_k\bigr)
\]
and set $s = i-j$, with $s\in\{-n+1,\ldots,-1,1,\ldots,n-1 \}$ since $i\neq j$.

Now it is time to specify the permutation $\pi$.  We choose $\pi$ such that $\pi(i) = i+1$ for $i<m$ and $\pi(m)=1$.  This is a permutation with no fixed point in the set $\{1,\ldots,m\}$ and thus $\tr(\pi) = 0$ since the permutation matrix has only zeros on the diagonal.  This is as well valid for all powers $\pi^s$ of the matrix as long as $s \neq m$  and $s\neq 0$.  This is true due to the restriction on $s$ above and since we assumed $\dim A \leq \dim B$ i.e. $n\leq m$.  Now, since $\pi^s$ has no fixed point in $\{1,\ldots,m\}$, it follows that
\[
\Tr \bigl( \pi^s\,\Lambda_k\bigr) = \sum_{\mu=1}^m \lambda_{i\mu} \langle \mu|\pi^s|\mu\rangle= 0,
\]
finally showing that $\Tr_B \Omega=0$.
\end{proof}

\end{document}